\documentclass[12pt, reqno]{amsart}

\usepackage[a4paper, 
            inner = 2.5cm,
            outer = 2.5cm, 
            top = 2.5cm, 
            bottom = 3cm]{geometry}

\usepackage{amssymb, amsmath, amsthm, amsxtra, yhmath}
\usepackage[colorlinks]{hyperref}

\newcommand{\be}{\begin{equation}}
\newcommand{\ee}{\end{equation}}

\newcommand{\bD}{\mathbb{D}}
\newcommand{\bN}{\mathbb{N}}
\newcommand{\bR}{\mathbb{R}}
\newcommand{\bZ}{\mathbb{Z}}
\newcommand{\bW}{\mathbb{W}}

\newcommand{\cA}{\mathcal{A}}
\newcommand{\cB}{\mathcal{B}}
\newcommand{\cC}{\mathcal{C}}
\newcommand{\cD}{\mathcal{D}}
\newcommand{\cH}{\mathcal{H}}
\newcommand{\cL}{\mathcal{L}}
\newcommand{\cO}{\mathcal{O}}
\newcommand{\cR}{\mathcal{R}}
\newcommand{\cS}{\mathcal{S}}
\newcommand{\cT}{\mathcal{T}}
\newcommand{\cY}{\mathcal{Y}}

\newcommand{\mfT}{\mathfrak{T}}
\newcommand{\mfX}{\mathfrak{X}}

\newcommand{\bszero}{{\boldsymbol{0}}}
\newcommand{\bsone}{{\boldsymbol{1}}}
\newcommand{\bsX}{{\boldsymbol{X}}}
\newcommand{\bsa}{{\boldsymbol{a}}}
\newcommand{\bsb}{{\boldsymbol{b}}}
\newcommand{\bschi}{{\boldsymbol{\chi}}}
\newcommand{\bskappa}{{\boldsymbol{\kappa}}}

\newcommand{\ri}{\mathrm{i}}
\newcommand{\dd}{\mathrm{d}}
\newcommand{\diag}{\mathrm{diag}}
\newcommand{\tr}{\mathrm{tr}}

\newcommand{\half}{{\frac{1}{2}}}
\newcommand{\halfbeta}{{\frac{\beta}{2}}}
\newcommand{\quarter}{{\frac{1}{4}}}
\newcommand{\quarterbeta}{{\frac{\beta}{4}}}

\newcommand{\PB}[2]{ \{ #1, #2 \} }
\newcommand{\OD}[2]{\frac{\dd #1}{\dd #2}}
\newcommand{\PD}[2]{\frac{\partial #1}{\partial #2}}
\newcommand{\SET}[2]{ \{\, #1 \, \mid \, #2 \,\}}
\newcommand{\BR}[1]{ \left( #1 \right) }
\newcommand{\ABS}[1]{ \left\vert #1 \right\vert }
\newcommand{\BLOCKMAT}[4]{\begin{bmatrix} #1 & #2 \\ #3 & #4 \end{bmatrix}}
\newcommand{\COMM}[2]{\left[ #1, #2 \right]}

\newcommand{\midand}{\quad \text{and} \quad}

\numberwithin{equation}{section}

\theoremstyle{plain}
\newtheorem{THEOREM}{Theorem}[section]
\newtheorem{LEMMA}[THEOREM]{Lemma}
\newtheorem{PROPOSITION}[THEOREM]{Proposition}
\newtheorem{COROLLARY}[THEOREM]{Corollary}

\begin{document}

\title[]{Lax matrices for a $1$-parameter subfamily of van~Diejen--Toda chains}

\author[]{B\'ela G\'abor Pusztai}

\address{
    Interdisciplinary Excellence Centre,
    Bolyai Institute, 
    University of Szeged \\
    Aradi v\'ertan\'uk tere 1, 
    H-6720 Szeged, Hungary
}

\email[]{gpusztai@math.u-szeged.hu}



\begin{abstract}
In this paper, we construct Lax matrices for certain relativistic open 
Toda chains endowed with a one-sided $1$-parameter boundary interaction. 
Built upon the Lax representation of the dynamics, an algebraic solution 
algorithm is also exhibited. To our best knowledge, this particular 
$1$-parameter subfamily of van Diejen--Toda chains has not been analyzed 
in earlier literature.
\end{abstract}

\maketitle
\tableofcontents

\section{Introduction}\label{SECTION:Intro}
As is known, the relativistic story of the Toda systems began with the 
seminal work of Ruijsenaars \cite{R90}, and has received lots of attention 
from the outset \cite{BR88, BR89, S90, S91, S96, S97, S03, S18}. For a 
comprehensive review of the early developments we recommend \cite{KMZ} 
and references therein. From our perspective we must pay particular 
attention to the work \cite{S90}, since this is the first paper that 
introduced non-trivial integrable deformations of the relativistic 
Toda systems. The deformed systems of type-I defined in \cite{S90} are 
characterized by $4$ parameters, the type-II lattices contain $2$ 
parameters, whereas the type-III deformations contain no couplings. The 
next major step in this direction is due to van Diejen, who derived 
integrable deformations of the relativistic Toda chains with $9$ coupling 
constants \cite{D94}, including Suris' deformed systems as special cases. 
Since the appearance of these integrable deformations we have been 
witnessing a continuous development at the quantum level 
\cite{D95, KT, E, S, C, DE}. 

However, at the classical level the theory of the deformed relativistic 
Toda lattices still appears to be in its infancy. The only reasonable 
explanation of this situation is the lack of Lax representation of the 
dynamics for the most generic deformations. Indeed, to our best knowledge, 
so far only two papers have addressed the construction of Lax matrices for 
certain special variants of the deformed models. Essentially by a folding 
procedure, Ruijsenaars \cite{R90} derived $C$-type and $BC$-type models 
from the translational invariant open lattices associated with the $A$-type 
root systems. Though Ruijsenaars' folded models naturally inherit Lax pairs 
from the translational invariant lattices, they contain no (essential) 
coupling parameters. The most satisfactory results in this direction can be 
found in \cite{S90}, where spectral parameter dependent $2 \times 2$ Lax 
matrices are provided for Suris' type-I/II/III deformations. Furthermore, 
the underlying $r$-matrix structure is also exhibited. Still, there is a 
huge gap between the $4$-parameter family of type-I systems and the 
$9$-parameter family of the van Diejen--Toda chains. 

Although the full picture with the maximal number of coupling constants 
is still out of reach, in this paper we get one step closer to the solution
of this long-standing open problem by providing a detailed analysis for a 
particular $1$-parameter subfamily of van Diejen--Toda chains, that has 
not been studied earlier. In order to describe these chains, take an 
arbitrary integer $n \geq 3$ and consider the phase space
\be\label{P}
    P = \SET{\zeta = (\xi, \eta)}
            {\xi_1, \ldots, \xi_n, \eta_1, \ldots, \eta_n \in \bR} 
    = \bR^{2 n}
\ee
equipped with its standard smooth manifold structure. Also, we introduce 
a global coordinate system on it by the family of functions 
\be\label{q&theta}
    q_a(\zeta) = \xi_a
    \midand
    \theta_a(\zeta) = \eta_a
    \qquad
    (\zeta = (\xi, \eta) \in P, \; 1 \leq a \leq n).
\ee
As customary in the theory of the aforementioned relativistic integrable 
many-body systems, the coordinates $q_a$ and $\theta_a$ are called the 
particle positions and the particle rapidities, respectively. Regarding 
$P$ as a model of the cotangent bundle of $\bR^n$, it naturally carries 
the symplectic form
\be\label{omega}
    \omega = \sum_{c = 1}^n \dd q_c \wedge \dd \theta_c,
\ee
and the corresponding Poisson bracket reads
\be\label{PB}
    \PB{f}{h} 
    = \sum_{c = 1}^n 
        \BR{\PD{f}{q_c} \PD{h}{\theta_c} - \PD{f}{\theta_c} \PD{h}{q_c}}
    \qquad
    (f, h \in C^\infty(P)).
\ee

As concerns the inter-particle interaction of our integrable chains, it 
proves handy to introduce the single variable smooth function
\be\label{f_mu}
    f_\mu \colon \bR \rightarrow (0, \infty),
    \quad
    r \mapsto f_\mu(r) = \sqrt{1 + \mu^2 e^{-r}},
\ee
where $\mu$ is a real parameter. In passing we remark that for $\mu \neq 0$ 
we have
\be\label{f_mu_ineq}
    f_\mu(r) > 1.
\ee
Now, let $\beta \in (0, \infty)$ and $\kappa \in \bR$ be arbitrary constants
and consider the Hamiltonian
\be\label{H}
\begin{split}
    H_\kappa = 
    & \cosh(\beta \theta_1) f_\beta(q_1 - q_2)
        + \sum_{c = 2}^{n - 1} 
            \cosh(\beta \theta_c) 
            f_\beta(q_{c - 1} - q_c) f_\beta(q_c - q_{c + 1}) \\
    & +\cosh(\beta \theta_n) 
        f_\beta(q_{n - 1} - q_n) f_\beta(2 q_n) 
        f_{\kappa \beta}(2 q_n)
        + \kappa \beta^2 e^{-2 q_n}
        + \half \kappa \beta^4 e^{-(q_{n - 1} + q_n)}.
\end{split}
\ee
Since we shall keep fixed the `inverse speed of light' $\beta$ throughout 
the paper, it should not cause confusion that we suppress the dependence 
of the Hamiltonian on this parameter. However, parameter $\kappa$ does 
play a decisive role in our discussion. In this respect our first trivial 
observation is that for $\kappa < 0$ the Hamiltonian $H_\kappa$ is 
\emph{not} bounded from below. As numerical experiments show, in this 
case even the completeness of the Hamiltonian flow is questionable, and 
so the techniques we wish to present in this paper would require serious 
modifications. Therefore, from now on we \emph{assume}
\be\label{kappa_nonnegative}
    \kappa \geq 0.
\ee
Recalling \eqref{f_mu_ineq}, it is plain that this condition entails 
the lower estimate $H_\kappa > n$. It is also clear that the classical 
mechanical system $(P, \omega, H_\kappa)$ belongs to the class of 
deformed relativistic Toda systems introduced by van Diejen \cite{D94}. 
In order to clarify this connection, let us call to mind the Hamiltonian 
given in equation (37) of \cite{D94}, which depends on nine (complex) 
parameters: $g, g_0, g_0', g_1, g_1', k_0, k_0', k_1, k_1'$. Now, one 
can easily verify that by setting
\begin{align}
    & g = \beta, 
        \quad 
        g_0 = g_0' = g_1 = g_1' = 0, \\
    & k_0 = e^{\ri \pi/4} \sqrt{\beta},
        \quad
        k_0' = e^{-\ri \pi/4} \sqrt{\beta}, 
        \quad
        k_1 = e^{\ri \pi/4} \sqrt{\kappa \beta},
        \quad
        k_1' = e^{-\ri \pi/4} \sqrt{\kappa \beta},
\end{align}
van Diejen's multi-parametric Hamiltonian reduces to $\cR^* H_\kappa$, 
where $\cR$ is the canonical transformation
\be\label{cR}
    \cR \colon P \rightarrow P,
    \quad
    (\xi, \eta) \mapsto (-\xi, -\eta).
\ee
The appearance of the reflection map $\cR$ is due to the exponent $-r$ 
in the definition of function $f_\mu$~\eqref{f_mu}. Indeed, the majority 
of the papers, including \cite{D94}, apply the plus sign convention in 
the exponent. In any case, we see that $H_\kappa$~\eqref{H} describes 
an open relativistic Toda chain endowed with a one-sided $1$-parameter 
boundary interaction.

Having identified our models in the family of the van Diejen--Toda chains,
we are now in a position to make a comparison with the deformed models
appearing in \cite{R90, S90}. Looking at equation (6.24) in \cite{R90},
from \eqref{H} it is straightforward to see that in the special case 
$\kappa = 0$ the composition $\cR^* H_0$ reproduces Ruijsenaars' $C_n$-type 
Hamiltonian. On the other hand, the comparison with the deformed models 
defined in \cite{S90} is a bit more subtle. Nevertheless, the type-III 
lattices can be ruled out immediately, since they have boundary interactions 
on both sides. By inspecting the type-I/II deformations, one may 
recognize immediately that only the type-I lattices bear a resemblance to 
our models $H_\kappa$~\eqref{H}, but exact identification can be made 
only in the special case $\kappa = 0$. Indeed, $H_0$ corresponds to Suris' 
$C_n$-type lattice characterized by the parameter set 
$(a_1, b_1, a_n, b_n) = (0, 0, 0, 1)$ as given below equation (6) in 
\cite{S90}. Incidentally, this particular lattice coincides with
Ruijsenaars' $C_n$-type model. To sum up, it is safe to say that Lax 
matrices related to our chains $H_\kappa$~\eqref{H} have appeared in the 
literature only for $\kappa = 0$, and so the real novelty of our paper 
is the construction of Lax representation to the $\kappa > 0$ case. In 
passing we mention that at a more appropriate place of the paper, in 
subsection \ref{SUBSECTION:new_variables}, we perform a detailed comparison 
with Suris' deformed models.

Though the dependence of $H_\kappa$~\eqref{H} on $\kappa$ is clearly visible, 
we still have to convince ourselves that different values of the parameter 
lead to `essentially different' Hamiltonians. Being cautious is not without 
reasons: it has been observed even at the level of the non-relativistic Toda 
systems that by shifting the particle positions, one can introduce `fake' 
coupling constants into the models. To illustrate this phenomenon in the 
context of our systems, for any 
$\bschi = (\chi_1, \ldots, \chi_n) \in (0, \infty)^n$ define the canonical 
transformation
\be\label{mfT}
    \mfT_\bschi \colon P \rightarrow P,
    \quad
    (\xi_1, \ldots, \xi_n, \eta_1, \ldots, \eta_n)
    \mapsto
    (\xi_1 + \ln(\chi_1), \ldots, \xi_n + \ln(\chi_n), 
        \eta_1, \ldots, \eta_n).
\ee
Notice that by composing $H_\kappa$~\eqref{H} with $\mfT_\bschi$ we obtain
\be\label{H&mfT}
\begin{split}
    \mfT_\bschi^* H_\kappa = 
    & \cosh(\beta \theta_1) f_{\beta (\chi_2 / \chi_1)^\half}(q_1 - q_2) \\
    & + \sum_{c = 2}^{n - 1} 
            \cosh(\beta \theta_c) 
            f_{\beta (\chi_c / \chi_{c - 1})^\half}(q_{c - 1} - q_c) 
            f_{\beta (\chi_{c + 1} / \chi_c)^\half}(q_c - q_{c + 1}) \\
    & + \cosh(\beta \theta_n) 
            f_{\beta (\chi_n / \chi_{n - 1})^\half}(q_{n - 1} - q_n) 
            f_{\beta / \chi_n}(2 q_n) 
            f_{\kappa \beta / \chi_n}(2 q_n) \\
    & + \kappa \beta^2 \chi_n^{-2} e^{-2 q_n}
        + \half \kappa \beta^4 \chi_{n - 1}^{-1} \chi_n^{-1} 
            e^{-(q_{n - 1} + q_n)},
\end{split}
\ee
leading to the proliferation of the coupling constants in a trivial 
manner. Working backward, this observation can also be utilized to 
investigate whether parameter $\kappa$ can be eliminated: given 
$\kappa_1, \kappa_2 \in [0, \infty)$, we say that Hamiltonians 
$H_{\kappa_1}$ and $H_{\kappa_2}$ are \emph{equivalent up to change 
of couplings}, in notation $H_{\kappa_1} \sim H_{\kappa_2}$, if there 
is an $n$-tuple $\bschi \in (0, \infty)^n$ such that 
$H_{\kappa_2} = \mfT_\bschi^* H_{\kappa_1}$. It is evident that $\sim$
is an equivalence relation on the set of our distinguished Hamiltonians
\eqref{H}. Exploiting the explicit formula \eqref{H&mfT} one can also 
verify that
\be\label{param_equiv}
    H_{\kappa_1} \sim H_{\kappa_2}
    \text{ \emph{if and only if} 
            $\kappa_1 = \kappa_2$ or $\kappa_1 \kappa_2 = 1$.}
\ee
This characterization of the equivalence may seem a bit strange at 
first. However, if $\kappa > 0$, then with the special $n$-tuple 
$\bskappa = (\kappa, \ldots, \kappa)$ from \eqref{H&mfT} we do get
\be
    \mfT_\bskappa^* H_\kappa = H_{1 / \kappa}, \text{ and so } 
    H_{\kappa} \sim H_{1 / \kappa}.
\ee
The content of the above line can be interpreted by saying that the small
coupling and the strong coupling regimes are in duality. Utilizing this
duality, we could safely assume that $0 \leq \kappa \leq 1$. Moreover, 
if $\kappa_1, \kappa_2 \in [0, 1]$ and $\kappa_1 \neq \kappa_1$, then by 
\eqref{param_equiv} we have $H_{\kappa_1} \nsim H_{\kappa_2}$, implying 
that we cannot get rid of parameter $\kappa$ from $H_\kappa$~\eqref{H} by 
shifting the particle coordinates. In this sense, $\kappa$ turns out to 
be an essential parameter.

In the light of the above discussion, one may wonder why we incorporate 
the inverse speed of light into the couplings of the Hamiltonian 
$H_\kappa$~\eqref{H}. Indeed, with the special $n$-tuple 
$\bschi_\beta = (\beta^{2 n - 1}, \ldots, \beta^3, \beta)$, from 
\eqref{H&mfT} one sees immediately that
\be
\begin{split}
    \mfT_{\bschi_\beta}^* H_\kappa =
    & \cosh(\beta \theta_1) f_1(q_1 - q_2)
        + \sum_{c = 2}^{n - 1} 
            \cosh(\beta \theta_c) 
            f_1(q_{c - 1} - q_c) f_1(q_c - q_{c + 1}) \\
    & +\cosh(\beta \theta_n) 
        f_1(q_{n - 1} - q_n) f_1(2 q_n) f_{\kappa}(2 q_n)
        + \kappa e^{-2 q_n} + \half \kappa e^{-(q_{n - 1} + q_n)},
\end{split}
\ee
and so $\beta$ disappears from the expressions governing the particle 
interactions. In spite of this apparent simplification, we keep $\beta$ 
for two reasons. First, it makes easier to control the transition from 
the relativistic models to their non-relativistic counterparts. 
Specifically, by expanding $H_\kappa$~\eqref{H} in $\beta$, one finds 
\be
    H_\kappa = n + \beta^2 H_\kappa^{\text{n.r.}} + \cO(\beta^4),
\ee
where
\be\label{H_n.r._kappa}
    H_\kappa^{\text{n.r.}} 
    = \half \sum_{c = 1}^n \theta_c^2
        + \sum_{c = 1}^{n - 1} e^{-(q_c - q_{c + 1})} 
        + \half (1 + \kappa)^2 e^{-2 q_n}
\ee
can be identified with the Hamiltonian of the $C_n$-type non-relativistic
open Toda lattice. Note that with the special $n$-tuple
$\bschi_\kappa = (1 + \kappa, \ldots, 1 + \kappa)$ the shift defined in 
\eqref{mfT} yields
\be\label{H_n.r.}
    \mfT_{\bschi_\kappa}^* H_\kappa^{\text{n.r.}} 
    = \half \sum_{c = 1}^n \theta_c^2
        + \sum_{c = 1}^{n - 1} e^{-(q_c - q_{c + 1})} 
        + \half e^{-2 q_n},
\ee
thus parameter $\kappa$ can be eliminated from \eqref{H_n.r._kappa}, as
expected. So, contrary to $H_\kappa$~\eqref{H}, the coupling parameter in
$H^\text{n.r.}_\kappa$ is purely artificial. In this sense the one-sided 
$1$-parameter boundary interaction given in $H_\kappa$~\eqref{H} has no 
non-trivial footprint on the non-relativistic level. Admittedly, the second 
reason for keeping $\beta$ has its roots in wishful thinking. In \cite{S90} 
Suris observed a remarkable relationship between certain discrete time 
generalized Toda lattices and Ruijsenaars' relativistic Toda lattices. An
essential ingredient of this picture is the correspondence between the 
discrete time step-size and the inverse speed of light. Thus, by emphasizing 
the role of $\beta$ in our paper, our hope is that it may 
facilitate to find a \emph{discrete time} generalized Toda lattice 
interpretation of the deformed relativistic Toda model $H_\kappa$~\eqref{H}.

Having described the integrable systems of our interest, now we wish to 
briefly outline the content of the paper. Since $H_\kappa$~\eqref{H} is 
not a textbook Hamiltonian, in Section~\ref{SECTION:Analyzing_the_dynamics} 
we investigate the dynamics, with particular emphasis on the completeness
of the Hamiltonian flow. Also, motivated by the work of Suris \cite{S90}, 
in this section we write down the Hamiltonian equations of motion in a 
distinguished set of canonical variables. Incidentally, in this Darboux 
system the relationship between a particular instance of Suris' generalized 
Toda lattice of type-I and the deformed relativistic Toda model 
$H_\kappa$~\eqref{H} with $\kappa = 0$ also becomes transparent, as 
discussed below Proposition~\ref{LEMMA:H}. The ultimate goal of 
Section~\ref{SECTION:Lax matrices} is to construct Lax matrices for the 
dynamics generated by $H_\kappa$~\eqref{H}. First, as formulated in 
Theorem~\ref{THEOREM:Lax_triad}, we set up a so-called `Lax triad' for 
the dynamics, which can be seen as a weaker form of the Lax equation. 
Built upon this intermediate step, an honest Lax representation of the 
dynamics also emerges, as summarized in Theorem~\ref{THEOREM:Lax_eqn}. 
The members of the proposed Lax pair $(\cL, \cA)$ are defined in \eqref{cL} 
and \eqref{cA}, respectively. It is worth mentioning that in the most 
interesting (new) cases, that is for $\kappa > 0$, the Lax matrix $\cL$ 
has lower bandwidth $2$, whereas $\cA$ is pentadiagonal. By exploring 
further the relationships among $H_\kappa$, $\cL$ and $\cA$, in 
Section~\ref{SECTION:Solution_algo} we provide an algebraic solution 
algorithm for the Hamiltonian dynamics generated by $H_\kappa$. As 
can be seen in Theorem~\ref{THEOREM:Solution_algo}, and in the subsequent 
discussion, the time evolution of the particle coordinates can be recovered 
from the $LDU$ factorization of certain exponential matrix flow constructed 
with the aid of the Lax matrix $\cL$. Finally, in Section
\ref{SECTION:Discussion} we discuss some open problems related to the
dynamical system \eqref{H}.

\section{Analyzing the dynamics}\label{SECTION:Analyzing_the_dynamics}
In the first half of this section we address the issue of completeness 
of the Hamiltonian flow generated by $H_\kappa$~\eqref{H}. Our analysis 
hinges on a time reversal argument, which is a well-known technique to 
the experts of the area. Nevertheless, since completeness plays a crucial 
role in our investigations, we present this material, too, in a concise 
manner.

In the second half of the section we introduce a distinguished set of 
canonical variables, that will pave the way to the construction of Lax 
matrices. Let us also note that starting from this section we shall keep 
the non-negative parameter $\kappa$ fixed, and shall apply the shorthand 
notation $H = H_\kappa$.

Finally, a further piece of notation: with any strictly positive integer 
$m \in \bN$ we associate the finite subset
\be\label{bN_m}
    \bN_m = \{ \, 1, \ldots, m \, \} \subseteq \bN.
\ee

\subsection{Completeness of the flow}\label{SUBSECTION:Completeness}
Recalling the phase space $P$~\eqref{P}, consider the smooth map
\be\label{T}
    T \colon P \rightarrow P,
    \quad
    (\xi, \eta) \mapsto (\xi, -\eta). 
\ee
Since $T$ is an involution, it is automatically invertible with inverse
$T^{-1} = T$. Remembering the standard coordinates \eqref{q&theta}, it is 
evident that
\be\label{T^*}
    T^* q_a = q_a
    \midand
    T^* \theta_a = -\theta_a
    \qquad
    (a \in \bN_n).
\ee
To put it simple, $T$ reverses the rapidities. Since the $\cosh$ function 
is even, the Hamiltonian $H$~\eqref{H} is invariant under the reversal of 
the rapidities; that is, $T^* H = H$. Notice also that for the pullback of 
the symplectic form \eqref{omega} by $T$ we get 
\be\label{T&omega}
    T^* \omega = -\omega,
\ee
meaning that $T$ is actually an anti-symplectomorphism. 

To proceed, with the aid of the Poisson bracket \eqref{PB} we also introduce
the Hamiltonian vector field $\bsX_H \in \mfX(P)$ corresponding to our 
distinguished Hamiltonian function $H$~\eqref{H}. Namely, for its action 
on the family of smooth functions we employ the usual convention
\be\label{bsX_H}
    \bsX_H[f] = \PB{f}{H}
    \qquad
    (f \in C^\infty(P)).
\ee
Taking an arbitrary point $\zeta \in P$, let
\be\label{gamma_zeta}
    \gamma_\zeta \colon (\bsa_\zeta, \bsb_\zeta) \rightarrow P,
    \quad
    t \mapsto \gamma_\zeta(t)
\ee
be the (unique) maximally defined integral curve of $\bsX_H$ satisfying 
the initial condition
\be\label{IC}
    \gamma_\zeta(0) = \zeta.
\ee
That is, the curve is determined by the differential equation
\be
    \dot{\gamma}_\zeta(t) = (\bsX_H)_{\gamma_\zeta(t)}
    \qquad
    (t \in (\bsa_\zeta, \bsb_\zeta)).
\ee
Here and below, the dot refers to the differentiation with respect to time. 
Note also that the endpoints of the maximal domain are appropriate (unique) 
constants obeying
\be\label{endpoints}
    -\infty \leq \bsa_\zeta < 0 < \bsb_\zeta \leq \infty.
\ee

Now, making use of $T$~\eqref{T}, for any $\zeta \in P$ consider the 
well-defined smooth curve
\be\label{curve_c}
    c_\zeta \colon (-\bsb_\zeta, -\bsa_\zeta) \rightarrow P,
    \quad
    t \mapsto c_\zeta(t) = T(\gamma_\zeta(-t)).
\ee
Since $H$ is invariant under the anti-Poisson map $T$, it is plain that for 
the action of the tangent vector $\dot{c}_\zeta(t) \in T_{c_\zeta(t)} P$ on 
any smooth function $f \in C^\infty(P)$ we can write
\be\label{c_dot}
\begin{split}
    & \dot{c}_\zeta(t)[f] 
        = - (T \circ \gamma_\zeta)\spdot(-t)[f]
        = - \dot{\gamma}_\zeta(-t) [f \circ T] \\
    & \quad = -(\bsX_H [f \circ T])(\gamma_\zeta(-t))
            = - \PB{f \circ T}{H}(\gamma_\zeta(-t)) \\
    & \quad = - \PB{f \circ T}{H \circ T}(\gamma_\zeta(-t))
            = \PB{f}{H}(T(\gamma_\zeta(-t))) \\
    & \quad = (\bsX_H [f])(c_\zeta(t)) 
            =(\bsX_H)_{c_\zeta(t)}[f].
\end{split}
\ee
That is, $c_\zeta$ in an integral curve of $\bsX_H$ satisfying the initial 
condition
\be\label{c(0)}
    c_\zeta(0) = T(\gamma_\zeta(-0)) = T(\zeta).
\ee
Thus, simply by comparing $c_\zeta$ with the maximal integral curve
\be\label{gamma_T(zeta)}
    \gamma_{T(\zeta)} 
    \colon (\bsa_{T(\zeta)}, \bsb_{T(\zeta)}) \rightarrow P,
    \quad
    t \mapsto \gamma_{T(\zeta)}(t),
\ee
from the maximality of $\gamma_{T(\zeta)}$ we infer that
\be\label{comparing_domains}
    (-\bsb_\zeta, -\bsa_\zeta) \subseteq (\bsa_{T(\zeta)}, \bsb_{T(\zeta)}),
\ee
and also
\be\label{c&gamma}
    c_\zeta(t) = \gamma_{T(\zeta)}(t)
    \qquad
    (t \in (-\bsb_\zeta, -\bsa_\zeta)).
\ee
Notice that, on account of \eqref{comparing_domains}, for all $\zeta \in P$
we have
\be\label{a&b_1}
    \bsa_{T(\zeta)} \leq -\bsb_\zeta 
    \midand 
    -\bsb_{T(\zeta)} \leq \bsa_\zeta.
\ee 
However, since $T$ is an involution, utilizing the above inequalities we can 
also write
\be\label{a&b_2}
    \bsa_\zeta = \bsa_{T(T(\zeta))} \leq -\bsb_{T(\zeta)}
    \midand 
    -\bsb_\zeta = -\bsb_{T(T(\zeta))} \leq \bsa_{T(\zeta)}.
\ee 
Now, keeping in mind \eqref{curve_c} and \eqref{c&gamma}, the comparison of 
\eqref{a&b_1} and \eqref{a&b_2} leads to the following result immediately.

\begin{LEMMA}\label{LEMMA:T-invariance}
Due to the invariance property $T^* H = H$, for each point $\zeta \in P$ we 
have
\be\label{a&b_OK}
    \bsa_{T(\zeta)} = -\bsb_\zeta
    \midand
    \bsb_{T(\zeta)} = -\bsa_\zeta.
\ee
Also, for all $t \in (\bsa_{T(\zeta)}, \bsb_{T(\zeta)})$ we can write
$\gamma_{T(\zeta)}(t) = T(\gamma_\zeta(-t))$.
\end{LEMMA}

As the Lemma suggests, $T$~\eqref{T} is sometimes called the time reversal 
map. In any case, the message of the above result is crystal clear: 
if one wishes to analyze the Hamiltonian flow generated by a $T$-invariant 
Hamiltonian function, any property of the backward flow can be inferred 
effortlessly from the study of the forward flow. This observation is of 
completely general nature, since it applies to any Hamiltonian invariant 
under the reversal of rapidities. Of course, in order to derive sharper 
results on the flow, one has to face the peculiarities of the Hamiltonian 
at hand. With this in mind, we shall need estimates on the time evolution 
of the coordinate functions \eqref{q&theta}. 

Starting with the particle positions $q_a$ $(a \in \bN_n)$, from \eqref{PB} 
and \eqref{bsX_H} we find
\be\label{bsX_H_q_a_GEN}
    \bsX_H[q_a] = \PB{q_a}{H} = \PD{H}{\theta_a}.
\ee
Thus, giving a glance at $H$~\eqref{H}, it is straightforward that
\begin{align}
    \bsX_H[q_1] & = \beta \sinh(\beta \theta_1) f_\beta(q_1 - q_2), 
    \label{bsX_H_q_1__in_q&theta} \\
    \bsX_H[q_n] & = \beta \sinh(\beta \theta_n) 
                    f_\beta(q_{n - 1} - q_n) 
                    f_\beta(2 q_n) f_{\kappa \beta}(2 q_n),
    \label{bsX_H_q_n__in_q&theta}
\end{align}
whereas for $2 \leq a \leq n - 1$ we have
\be\label{bsX_H_q_a__in_q&theta}
    \bsX_H[q_a] = \beta \sinh(\beta \theta_a) 
                    f_\beta(q_{a - 1} - q_a) f_\beta(q_a - q_{a + 1}).
\ee
Therefore, combining the estimate \eqref{f_mu_ineq} with the trivial 
inequality
\be\label{sinh_ineq}
    \ABS{\sinh(r)} = \sinh(\ABS{r}) < \cosh(r)
    \qquad
    (r \in \bR), 
\ee
we conclude 
\be\label{bsX_H_q_ineq}
    \ABS{\bsX_H[q_a]} < \beta H
    \qquad
    (a \in \bN_n).
\ee

Turning to the rapidities $\theta_a$ $(a \in \bN_n)$, from the explicit form 
of $H$~\eqref{H} it is plain that
\be\label{ineq_1}
    e^{\ABS{\beta \theta_a}} < 2 \cosh(\beta \theta_a) < 2 H, 
\ee
which entails
\be
\label{theta_ineq}
    \ABS{\theta_a} < \beta^{-1} \ln(2 H).
\ee
Note that in the derivation of both inequalities \eqref{bsX_H_q_ineq} and 
\eqref{theta_ineq} it is critical that $\kappa$ is a non-negative parameter, 
as we imposed in \eqref{kappa_nonnegative}. Now, by exploiting these 
estimates, the following result is immediate.

\begin{THEOREM}\label{THEOREM:completeness}
The Hamiltonian vector field $\bsX_H$~\eqref{bsX_H} corresponding to the 
Hamiltonian function $H$~\eqref{H} is complete.
\end{THEOREM}

\begin{proof}
According to Lemma~\ref{LEMMA:T-invariance}, it suffices to show that the
Hamiltonian flow generated by $H$ is forward-complete. In other words, we 
must show that for each $\zeta \in P$ we have $\bsb_\zeta = \infty$. 

Suppose to the contrary that there is a point $\zeta \in P$ such that 
$\bsb_\zeta < \infty$. Since $H$ is a first integral of $\bsX_H$, on account 
of \eqref{bsX_H_q_ineq} it is obvious that $\forall a \in \bN_n$ and 
$\forall t \in [0, \bsb_\zeta)$ we can write
\be\label{completeness_ineq_1}
    \ABS{(q_a \circ \gamma_\zeta)\spdot(t)}
        = \vert (\bsX_H)_{\gamma_\zeta(t)}[q_a] \vert
        < \beta H(\gamma_\zeta(t))
        = \beta H(\zeta),
\ee
from where we infer the estimate
\be\label{completeness_ineq_2}
\begin{split}
    & \ABS{q_a(\gamma_\zeta(t)) - q_a(\zeta)} 
        = \ABS{q_a(\gamma_\zeta(t)) - q_a(\gamma_\zeta(0))} \\
    & \quad 
        = \ABS{ \int_0^t (q_a \circ \gamma_\zeta)\spdot(\tau) \, \dd \tau}
        \leq \int_0^t \ABS{(q_a \circ \gamma_\zeta)\spdot(\tau)} \, \dd \tau \\
    & \quad 
        \leq \int_0^t \beta H(\zeta) \, \dd \tau
        = \beta H(\zeta) t < \beta H(\zeta) \bsb_\zeta.
\end{split}
\ee
That is, upon introducing the closed cube
\be\label{cube_B}
    \cB_\zeta = \SET{\xi \in \bR^n}
                    {\forall a \in \bN_n 
                        \text{ we have } 
                        \ABS{\xi_a - q_a(\zeta)} 
                        \leq \beta H(\zeta) \bsb_\zeta},
\ee
we can write
\be\label{q_in_cube}
    (q_1(\gamma_\zeta(t)), \ldots, q_n(\gamma_\zeta(t))) \in \cB_\zeta
    \qquad
    (t \in [0, \bsb_\zeta)).
\ee

As concerns the time evolution of $\theta_a$ $(a \in \bN_n)$, from 
\eqref{theta_ineq} it is plain that $\forall t \in [0, \bsb_\zeta)$ we have
\be\label{completeness_ineq_3}
    \ABS{\theta_a(\gamma_\zeta(t))} 
        < \beta^{-1} \ln(2 H(\gamma_\zeta(t)))
        = \beta^{-1} \ln(2 H(\zeta)).
\ee
That is, with the aid of the closed cube
\be\label{cube_C}
    \cC_\zeta = \SET{\eta \in \bR^n}
                    {\forall a \in \bN_n 
                        \text{ we have } 
                        \ABS{\eta_a} \leq \beta^{-1} \ln(2 H(\zeta))},
\ee
the content of \eqref{completeness_ineq_3} can be rephrased as
\be\label{theta_in_cube}
    (\theta_1(\gamma_\zeta(t)), \ldots, \theta_n(\gamma_\zeta(t))) 
        \in \cC_\zeta
    \qquad
    (t \in [0, \bsb_\zeta)).
\ee
Combining this observation with \eqref{q_in_cube}, we conclude
\be\label{gamma_in_cube}
    \gamma_\zeta(t) \in \cB_\zeta \times \cC_\zeta
    \qquad
    (t \in [0, \bsb_\zeta)).
\ee

On the other hand, since $\bsb_\zeta < \infty$, the theory of ordinary 
differential equations guarantees that $\gamma_\zeta$ `escapes' from the 
\emph{compact} subset $\cB_\zeta \times \cC_\zeta \subseteq P$ (see e.g.
\cite[Theorem 3.7]{Si}). More precisely, one can find a small positive 
number $\delta \in (0, \bsb_\zeta)$ such that 
\be\label{gamma_escapes}
    \gamma_\zeta(t) \notin \cB_\zeta \times \cC_\zeta
    \qquad
    (t \in (\bsb_\zeta - \delta, \bsb_\zeta)),
\ee
contradicting \eqref{gamma_in_cube}. As a consequence, the assumption 
$\bsb_\zeta < \infty$ must be rejected, and so the proof is complete.
\end{proof}

\subsection{New set of canonical variables}\label{SUBSECTION:new_variables}
It was observed during the early developments of the relativistic Toda 
chains that many calculations could be made simpler in variables 
different from the original coordinates provided by the particle positions 
and rapidities \cite{BR88, BR89}. From Suris' paper \cite{S90} it is also
clear that these special change of coordinates are instrumental in finding
the link between certain discrete time generalized Toda lattices and the
relativistic Toda chains. Therefore, taking the lead of Section 5 in 
\cite{S90}, we find it convenient to introduce a family of smooth functions 
as follows. Let
\begin{align}
    & s_1 = \beta^{-1} \ln(f_\beta(q_1 - q_2)), 
    \label{s_1} \\
    & s_n = \beta^{-1} 
            \BR{-\ln(f_\beta(q_{n - 1} - q_n)) 
            + \ln(f_\beta(2 q_n)) 
            + \ln(f_{\kappa \beta}(2 q_n))},
    \label{s_n}
\end{align}
whereas for $2 \leq a \leq n - 1$ we define
\be\label{s_a}
    s_a 
    = \beta^{-1} 
        \BR{-\ln(f_\beta(q_{a - 1} - q_a)) + \ln(f_\beta(q_a - q_{a + 1}))}.
\ee
Built upon the above auxiliary objects, we also define
\be\label{p_a}
    p_a = \theta_a + s_a \in C^\infty(P)
    \qquad
    (a \in \bN_n),
\ee
and the map
\be\label{S}
    S \colon P \rightarrow P,
    \quad
    \zeta 
        \mapsto 
        (q_1(\zeta), \ldots, q_n(\zeta), p_1(\zeta), \ldots, p_n(\zeta)).
\ee
As it turns out, these new functions \eqref{p_a} provide a convenient
framework to analyze the Hamiltonian system \eqref{H}. As an added bonus,
the comparison with \cite{S90} also becomes more transparent.

We initiate the study of the functions \eqref{p_a} with an elementary 
observation: by working out their first order partial derivatives, 
there is no difficulty in verifying that
\be
    \PD{s_b}{q_a} = \PD{s_a}{q_b}
    \qquad
    (a, b \in \bN_n),
\ee
and so the following is immediate.

\begin{PROPOSITION}\label{PROPOSITION:closed_1_form}
The differential $1$-form $\sum_{c = 1}^n s_c \, \dd q_c$ is closed.
\end{PROPOSITION}

Incidentally, since $P$~\eqref{P} is topologically trivial, the closed 
differential $1$-form featuring the above Proposition is automatically 
exact. In other words, there is a smooth function $\cS \in C^\infty(P)$, 
unique up to a constant, such that
\be\label{cS}
    \sum_{c = 1}^n s_c \, \dd q_c = \dd \cS.
\ee
However, from our perspective it is more important that the symplectic form
\eqref{omega} is also exact, and it can be written as
\be\label{omega_exact}
    \omega = - \dd \Big( \sum_{c = 1}^n \theta_c \, \dd q_c \Big).
\ee
Combining this observation with Proposition~\ref{PROPOSITION:closed_1_form},
one can easily verify following result.

\begin{LEMMA}\label{LEMMA:S}
The map $S$~\eqref{S} is a symplectomorphism.
\end{LEMMA}

\begin{proof}
Note that the auxiliary functions $s_a$ are independent of the rapidities. 
Thus, simply by inspecting the defining formula of $S$~\eqref{S}, we see 
that it is an invertible map with inverse
\be\label{S_inverse}
    S^{-1}(\zeta) = (q_1(\zeta), \ldots, q_n(\zeta), 
                        \theta_1(\zeta) - s_1(\zeta), 
                        \ldots, 
                        \theta_n(\zeta) - s_n(\zeta))
    \qquad
    (\zeta \in P).
\ee
Since both $S$ and $S^{-1}$ are smooth, $S$ is a diffeomorphism. 

Next, pulling back the coordinate functions \eqref{q&theta} by $S$, it is 
also evident that
\be\label{S*q&p}
    S^* q_a = q_a
    \midand
    S^* \theta_a = p_a
    \qquad
    (a \in \bN_n).
\ee
Consequently, taking into account \eqref{omega_exact}, 
Proposition~\ref{PROPOSITION:closed_1_form} entails
\be\label{S*omega}
    S^* \omega = -\dd \Big( \sum_{c = 1}^n p_c \, \dd q_c \Big)
    = \omega -\dd \Big( \sum_{c = 1}^n s_c \, \dd q_c \Big)
    = \omega,
\ee
which clearly shows that $S$ is a symplectic map.
\end{proof}

\begin{COROLLARY}
\label{COROLLARY:Darboux_coords}
The system of smooth functions $(q_1, \ldots, q_n, p_1, \ldots, p_n)$ 
provides global canonical coordinates on the phase space $P$.
\end{COROLLARY}

Besides the above Darboux system, in the following we shall use extensively 
a distinguished family of non-canonical variables, too. Namely, for each 
$a \in \bN_n$ we define
\be\label{x_a}
    x_a = e^{\beta p_a},
\ee
as well as
\be\label{w_a}
    w_a
    = \begin{cases}
        \beta^2 e^{-(q_a - q_{a + 1})}, & \text{if $a \in \bN_{n - 1}$;} \\
        \beta^2 e^{-2 q_n},             & \text{if $a = n$.}
    \end{cases}
\ee
Though it should not cause confusion that the dependence of the above 
functions on the parameters $\beta$ and $\kappa$ is not indicated, we 
wish to stress that these parameters are still part of our analysis. 
Finally, it proves also handy to introduce the shorthand notation
\be\label{Delta}
    \Delta = 1 + \kappa^2 w_n.
\ee

\begin{LEMMA}\label{LEMMA:H}
In terms of the variables \eqref{x_a} and \eqref{w_a}, the 
Hamiltonian function \eqref{H} acquires the form
\be\label{H_in_x&w}
\begin{split}
    H = &   \half x_1 
            + \half \sum_{c = 1}^{n - 1} (x_c^{-1} + x_{c + 1})(1 + w_c)    
            + \half x_n^{-1} (1 + w_n) \Delta 
            + \kappa w_n + \half \kappa w_{n - 1} w_n.
\end{split}
\ee
\end{LEMMA}

\begin{proof}
Utilizing \eqref{f_mu}, \eqref{w_a} and \eqref{Delta}, we see that
\be
    f_\beta(2 q_n) = (1 + w_n)^\half
    \midand
    f_{\kappa \beta}(2 q_n) = \Delta^\half,
\ee
whereas
\be
    f_\beta(q_a - q_{a + 1}) = (1 + w_a)^\half
    \qquad
    (a \in \bN_{n - 1}).
\ee
Recalling \eqref{s_1} and \eqref{s_n}, from \eqref{x_a} we obtain
\be
    e^{\beta \theta_1} = x_1 f_\beta(q_1 - q_2)^{-1}
    \midand
    e^{\beta \theta_n} = x_n f_\beta(q_{n - 1} - q_n) 
                            f_\beta(2 q_n)^{-1} 
                            f_{\kappa \beta}(2 q_n)^{-1},
\ee
while \eqref{s_a} leads to
\be
    e^{\beta \theta_a} = x_a f_\beta(q_{a - 1} - q_a) 
                            f_\beta(q_a - q_{a + 1})^{-1}
    \qquad
    (2 \leq a \leq n - 1).
\ee
Now, simply by plugging the above expressions into $H$~\eqref{H}, the 
desired formula \eqref{H_in_x&w} follows immediately.
\end{proof}

As we promised in the Introduction, at this point we are prepared to make a
thorough comparison with the deformed systems defined in \cite{S90}. Giving 
a glance at the definitions \eqref{x_a} and \eqref{w_a}, from \eqref{H_in_x&w} 
we find immediately that
\be\label{H_in_q&p}
\begin{split}
    H = 
    & \sum_{c = 1}^n \cosh(\beta p_c) 
        + \half \beta^2 \sum_{c = 1}^{n - 1} 
            e^{q_{c + 1} - q_c} 
            (e^{-\beta p_c} + e^{\beta p_{c + 1}}) \\
    & + \half \beta^2 e^{-2 q_n - \beta p_n} 
        \big(1 + \kappa^2 (1 + e^{-2 q_n})\big)
        + \kappa \beta^2 e^{-2 q_n} 
        + \half \kappa \beta^4 e^{-(q_{n - 1} + q_n)}.
\end{split}
\ee
Looking at the corresponding Hamiltonian $J_1$ at the end of Section 3 in 
\cite{S90}, it is clear that only the type-I deformations share similarity
with \eqref{H_in_q&p}. Indeed, by identifying the discrete step-size with 
the inverse speed of light, in our notations the Hamiltonian $J_1$ of the 
type-I deformations translates into
\be\label{J_1}
    \half J_1 = 
        \sum_{c = 1}^n \cosh(\beta p_c) 
        + \half \beta^2 \sum_{c = 1}^{n - 1} 
            e^{q_{c + 1} - q_c} 
            (e^{-\beta p_c} + e^{\beta p_{c + 1}})
        + \half J^- + \half J^+,
\ee
where
\begin{align}
    & J^- = a_1 \beta^2 e^{q_1} (1 + e^{\beta p_1}) 
            + b_1 \beta^2 e^{2 q_1 + \beta p_1}, 
    \label{J-} \\
    & J^+ = a_n \beta^2 e^{-q_n} (1 + e^{-\beta p_n})
            + b_n \beta^2 e^{-2 q_n - \beta p_n}.
    \label{J+}
\end{align}
In the above formulas the coupling parameters $a_1$, $b_1$, $a_n$, and 
$b_n$ control the strength of the interaction at the boundaries. Now, by 
comparing \eqref{H_in_q&p} with \eqref{J_1}, we see that they coincide if 
and only if $\kappa = 0$ and $(a_1, b_1, a_n, b_n) = (0, 0, 0, 1)$. 
Consequently, the Lax matrices introduced in \cite{S90} are applicable for
our models \eqref{H} only in the special case $\kappa = 0$, which case was 
analyzed in \cite{R90}, too. This observation fully justifies our effort 
in the $\kappa > 0$ cases. Apart from clarifying the connection with Suris' 
work \cite{S90}, our hope is that formula \eqref{H_in_q&p} may prove to be 
the first step toward the discrete time integrable system interpretation 
of \eqref{H}.

Turning back to our analysis, our next goal is to describe the dynamics 
using the global canonical coordinates given in 
Corollary~\ref{COROLLARY:Darboux_coords}. In particular, we need their 
derivatives along the Hamiltonian vector field $\bsX_H$~\eqref{bsX_H}. 
However, thanks their canonicity, for all $a \in \bN_n$ we can write
\be\label{bsX_H_q&p}
    \bsX_H[q_a] = \PD{H}{p_a}
    \midand
    \bsX_H[p_a] = -\PD{H}{q_a}.
\ee
That is, our task amounts to calculating the first order partial derivatives 
of the Hamiltonian $H$~\eqref{H_in_x&w} with respect to the new coordinate 
functions.

\begin{LEMMA}\label{LEMMA:bsX_H_q&p}
In terms of the functions defined in \eqref{x_a} and \eqref{w_a}, the
derivative of the particle position $q_a$ $(a \in \bN_n)$ along the 
Hamiltonian vector field $\bsX_H$ has the form
\be\label{bsX_H_q_a}
    \bsX_H[q_a] = \begin{cases}
        \halfbeta \BR{x_1 - x_1^{-1} (1 + w_1)},                        & 
            \text{if $a = 1$;} \\
        \halfbeta \BR{x_a (1 + w_{a - 1}) - x_a^{-1} (1 + w_a)},        & 
            \text{if $2 \leq a \leq n - 1$;} \\
        \halfbeta \BR{x_n (1 + w_{n - 1}) - x_n^{-1} (1 + w_n) \Delta}, & 
            \text{if $a = n$.}
    \end{cases}
\ee
As concerns the functions $p_a$ $(a \in \bN_n)$ defined in \eqref{p_a}, we 
have 
\be\label{bsX_H_p_1}
    \bsX_H[p_1] = \half (x_1^{-1} + x_2) w_1,
\ee 
whereas for $2 \leq a \leq n - 1$ we can write
\be\label{bsX_H_p_a}
    \bsX_H[p_a] = - \half (x_{a - 1}^{-1} + x_a) w_{a - 1}
                    + \half (x_a^{-1} + x_{a + 1}) w_a 
                    + \half \kappa w_{n - 1} w_n \delta_{a, n - 1},
\ee
and
\be\label{bsX_H_p_n}
\begin{split}
    \bsX_H[p_n] = 
    & - \half (x_{n - 1}^{-1} + x_n) w_{n - 1}
                + x_n^{-1} w_n (1 + \kappa^2 + 2 \kappa^2 w_n) \\ 
    & + 2 \kappa w_n + \half \kappa w_{n - 1} w_n.
\end{split}
\ee
\end{LEMMA}

\begin{proof}
Giving a glance at \eqref{x_a}, we see
\be\label{x_a_parc_der}
    \PD{x_c}{q_a} = 0
    \midand
    \PD{x_c}{p_a} = \beta x_c \delta_{c, a}
    \qquad
    (a, c \in \bN_n).
\ee
Turning to the functions defined in \eqref{w_a}, it is also plain that
$\forall a \in \bN_n$ we have
\be\label{w_a_parc_der}
    \PD{w_c}{q_a} = - (\delta_{c, a} - \delta_{c + 1, a}) w_c
    \midand
    \PD{w_c}{p_a} = 0
    \qquad
    (c \in \bN_{n - 1}),
\ee
whilst
\be\label{w_n_der}
    \PD{w_n}{q_a} = -2 \delta_{n, a} w_n
    \midand
    \PD{w_n}{p_a} = 0.
\ee
Thus, simply by exploiting the explicit expression given in 
\eqref{H_in_x&w}, the formulas featuring the Lemma can be worked out 
without any difficulty.
\end{proof}

Now an important remark is in order. Similarly to the translational invariant
cases \cite{BR88, BR89, S90}, from the above two Lemmas it is transparent that 
the functions introduced in \eqref{x_a}, \eqref{w_a} and \eqref{Delta} may 
greatly facilitate our calculations. Indeed, the Hamiltonian~\eqref{H_in_x&w} 
in Lemma~\ref{LEMMA:H}, as well as the derivatives featuring 
Lemma~\ref{LEMMA:bsX_H_q&p} are \emph{polynomial} expressions of $x_a$, 
$x_a^{-1}$, $w_a$ and $\Delta$. For this reason, in the rest of the paper we 
shall express almost all objects only in terms of these functions. With this 
in mind, we find it advantageous to record here the 
derivatives of \eqref{x_a} and \eqref{w_a} along the Hamiltonian vector field 
$\bsX_H$. As concerns $x_a$ $(a \in \bN_n)$, it is plain that
\be\label{x_a_der}
    \bsX_H[x_a] = \beta x_a \bsX_H[p_a].
\ee
Thus, thanks to the explicit expressions \eqref{bsX_H_p_1}, \eqref{bsX_H_p_a} 
and \eqref{bsX_H_p_n}, the control over the derivatives $\bsX_H[x_a]$ is 
complete. Turning to $w_a$ $(a \in \bN_n)$, notice that
\be\label{w_a_der}
    w_a^{-1} \bsX_H[w_a]
    = \begin{cases}
        - \bsX_H[q_a] + \bsX_H[q_{a + 1}], 
            & \text{if $a \in \bN_{n - 1}$;} \\
        - 2 \bsX_H[q_n],
            & \text{if $a = n$.} 
    \end{cases}
\ee
Therefore, utilizing~\eqref{bsX_H_q_a}, we can obtain formulas for these
derivatives as well. Since we shall encounter them quite often, below we 
make them explicit. First, clearly we have
\be\label{bsX_H_w_1}
    \bsX_H[w_1] 
    = \halfbeta w_1 
        \BR{ -x_1 + (x_1^{-1} + x_2) (1 + w_1) - x_2^{-1}(1 + w_2) }.
\ee
Next, for $2 \leq a \leq n - 2$ we can write
\be\label{bsX_H_w_a}
    \bsX_H[w_a] = \halfbeta w_a 
                    \BR{- x_a (1 + w_{a - 1})
                        + (x_a^{-1} + x_{a + 1}) (1 + w_a)
                        - x_{a + 1}^{-1}(1 + w_{a + 1})}, 
\ee
whereas
\be\label{bsX_H_w_n-1}
\begin{split}
    \bsX_H[w_{n - 1}] = 
    \halfbeta w_{n - 1} 
        \big( & - x_{n - 1} (1 + w_{n - 2})
                + (x_{n - 1}^{-1} + x_n) (1 + w_{n - 1}) \\ 
            & - x_n^{-1}(1 + w_n) \Delta \big),
\end{split}
\ee
and finally
\be\label{bsX_H_w_n}
    \bsX_H[w_n] 
    = \beta w_n \BR{ -x_n (1 + w_{n - 1}) + x_n^{-1}(1 + w_n) \Delta }.
\ee

\section{Lax matrices}\label{SECTION:Lax matrices}
Mainly for fixing our conventions, in the beginning of this section we gather 
some rudimentary facts from the theory of matrices. Subsequently, we associate
a Lax triad and a Lax equation to the dynamics governed by the Hamiltonian
\eqref{H}.

\subsection{Algebraic preliminaries}\label{SUBSECTION:algebraic_preliminaries}
Taking an arbitrary positive 
integer $m \in \bN$, let $\bszero_m$ and $\bsone_m$ denote the $m \times m$
zero matrix and the $m \times m$ unit matrix, respectively. Also, for all 
$a, b \in \bN_m$ consider the $m \times m$ elementary matrix $E_{a, b}$ 
defined with the entries
\be\label{E}
    (E_{a, b})_{c, d} = \delta_{a, c} \delta_{b, d}
    \qquad
    (c, d \in \bN_m).
\ee
As is known, they obey the commutation relations
\be\label{E_commut_rel}
    E_{a, b} E_{c, d} = \delta_{b, c} E_{a, d}
    \qquad
    (a, b, c, d \in \bN_m).
\ee
Next, in the real matrix algebra $\bR^{m \times m}$, for each $k \in \bZ$ 
we introduce a distinguished subspace
\be
    (\bR^{m \times m})_k = 
    \begin{cases}
        \text{span}
            \SET{E_{a, b}}{a, b \in \bN_m \text{ and } b - a = k},  &
            \text{if $\ABS{k} < m$;} \\
        \{ \bszero_m \},                                            &
            \text{if $\ABS{k} \geq m$.}
    \end{cases}
\ee
As a vector space decomposition, we clearly have
\be\label{decomp}
    \bR^{m \times m} = \bigoplus_{k \in \bZ} (\bR^{m \times m})_k
    = \bigoplus_{k = -m + 1}^{m - 1} (\bR^{m \times m})_k.
\ee
Regarding its compatibility with the matrix multiplication, from the above
commutation relations \eqref{E_commut_rel} it arises that
\be\label{subspaces&multiplication}
    (\bR^{m \times m})_k \cdot (\bR^{m \times m})_l
    \subseteq
    (\bR^{m \times m})_{k + l}
    \qquad
    (k, l \in \bZ),
\ee
so decomposition \eqref{decomp} is actually a $\bZ$-grading. Sometimes 
it is called the standard grading, or the principal gradation of 
the matrix algebra $\bR^{m \times m}$.

Note that, on account of \eqref{decomp}, each $m \times m$ matrix $X$ can be 
uniquely decomposed as
\be\label{X_decomp}
    X = \sum_{k \in \bZ} X_k = \sum_{k = -m + 1}^{m - 1} X_k, 
\ee
where $X_k \in (\bR^{m \times m})_k$. It is plain that $X_0$ is the diagonal 
part of $X$, and in general $X_k$ may have nontrivial entries only on its 
$\ABS{k}$th superdiagonal or subdiagonal, depending on the sign of $k$. In 
association with the $\bZ$-grading \eqref{decomp} we find it convenient to 
introduce the shorthand notations
\be\label{X_triang}
    X_{\leq 0} = \sum_{k \leq 0} X_k 
    \midand
    X_{\geq 0} = \sum_{k \geq 0} X_k
\ee
for the lower triangular and the upper triangular part of $X$, 
respectively. Relatedly, for the strictly lower triangular and 
the strictly upper triangular part of $X$ we shall write
\be\label{X_strictly_triang}
    X_{< 0} = \sum_{k < 0} X_k 
    \midand
    X_{> 0} = \sum_{k > 0} X_k.
\ee
As a further piece of preparation, let us recall that a square matrix is 
called a unit (lower or upper) triangular matrix, if it is (lower or upper) 
triangular with ones on the main diagonal.

We conclude this subsection with an important remark. In the rest of this 
section we shall employ the above notations only for $n \times n$ matrices. 
However, starting from Section~\ref{SECTION:Solution_algo}, we shall apply 
the notations \eqref{X_decomp}, \eqref{X_triang} and \eqref{X_strictly_triang}
exclusively for $(2 n) \times (2 n)$ matrices.

\subsection{Lax triad}\label{SUBSECTION:Lax_triad}
As we have already stipulated, till the end of this section the notations 
introduced in subsection~\ref{SUBSECTION:algebraic_preliminaries} are in 
effect only for quadratic matrices of size $n$. In particular, $E_{a, b}$ 
stands for the $n \times n$ elementary matrix featuring $1$ in the 
intersection of the $a$th row with the $b$th column and $0$ everywhere else. 
We shall also frequently encounter the $n \times n$ reversal matrix, which 
is the distinguished permutation matrix
\be\label{R}
    R 
    = \sum_{c = 1}^n E_{c, n + 1 - c}
    = \begin{bmatrix}
        0           & \cdots    & 0         & 1             \\
        0           & \cdots    & 1         & 0             \\
        \vdots      & \adots    & \vdots    & \vdots        \\
        1           & \cdots    & 0         & 0
    \end{bmatrix}.
\ee
Since $R^2 = \bsone_n$, the matrix $R$ is invertible, and
\be\label{R_inv}
    R^{-1} = R = R^T.
\ee
It is also evident that $\forall a, b \in \bN_n$ we have
\be\label{R&E}
    R E_{a, b} = E_{n + 1 - a, b}
    \midand
    E_{a, b} R = E_{a, n + 1 - b},
\ee
and so
\be\label{RER}
    R E_{a, b} R^{-1} = R E_{a, b} R = E_{n + 1 -a, n + 1 - b}.
\ee
Incidentally, by virtue of the above relations it is plain that 
$\forall k \in \bZ$ we have
\be\label{R&grading}
    R (\bR^{n \times n})_k R^{-1}
    = R (\bR^{n \times n})_k R
    = (\bR^{n \times n})_{-k}.
\ee
In particular, the $R$-conjugate of a lower triangular matrix is upper 
triangular, and vice versa.

Having all the necessary algebraic objects at our disposal, utilizing the 
functions defined in \eqref{x_a} and \eqref{Delta} we now build up the 
$n \times n$ invertible diagonal matrix
\be\label{bD}
    \bD 
    = \sum_{c = 1}^{n - 1} x_c^\half E_{c, c} 
        + x_n^\half \Delta^{-\quarter} E_{n, n}.
\ee
With the aid of its derivative along the Hamiltonian vector field 
$\bsX_H$~\eqref{bsX_H}, we also define
\be\label{cD}
    \cD = \bD^{-1} \bsX_H[\bD].
\ee
It is plain that
\be
    \cD = \half \sum_{c = 1}^{n - 1} x_c^{-1} \bsX_H[x_c] E_{c, c}
            + \Big( \half x_n^{-1} \bsX_H[x_n] 
                    - \quarter \Delta^{-1} \bsX_H[\Delta] \Big) E_{n, n}, 
\ee
so the first $n - 1$ diagonal entries of $\cD$ are under complete control,
due to \eqref{x_a_der} and the explicit formulas given in 
Lemma~\ref{LEMMA:bsX_H_q&p}. Indeed, straightforward calculations show that
\be\label{cD_1_1}
    \cD_{1, 1} = \quarterbeta (x_1^{-1} + x_2) w_1,
\ee
whereas for $2 \leq a \leq n - 1$ we have
\be\label{cD_a_a}
    \cD_{a, a} = \quarterbeta
                    \BR{-(x_{a - 1}^{-1} + x_a) w_{a - 1}
                        + (x_a^{-1} + x_{a + 1}) w_a
                        + \kappa w_{n - 1} w_n \delta_{a, n - 1}}.
\ee
However, the last diagonal entry of $\cD$ requires a more careful treatment. 
For this reason, we find it convenient to introduce the shorthand notations
\begin{align}
    & \Gamma = x_n(1 + w_{n - 1}) \Delta^{-1} + x_n^{-1} (1 + w_n), 
        \label{Gamma} \\
    & \Upsilon = - x_n(1 + w_{n - 1}) \Delta^{-1} + x_n^{-1} (1 + w_n).
        \label{Upsilon}
\end{align}
Now, looking back to \eqref{bsX_H_w_n}, we see that
\be\label{bsX_H_w_n_OK}
    \bsX_H[w_n] = \beta w_n \Delta \Upsilon,
\ee
and so from the explicit form of $\Delta$~\eqref{Delta} we deduce
that
\be\label{bsX_H_Delta}
    \Delta^{-1} \bsX_H[\Delta] = \beta \kappa^2 w_n \Upsilon.
\ee
Therefore, giving a glance at \eqref{x_a_der} and Lemma~\ref{LEMMA:bsX_H_q&p}, 
for the remaining diagonal entry of $\cD$~\eqref{cD} we can cook up the 
formula
\be\label{cD_n_n}
    \cD_{n, n} = \frac{\beta}{4} 
                    \BR{-(x_{n - 1}^{-1} + x_n) w_{n - 1} 
                        + 2 x_n^{-1} w_n \Delta 
                        + 4 \kappa w_n + \kappa w_{n - 1} w_n 
                        + \kappa^2 w_n \Gamma}.
\ee

To proceed, let us also introduce the $n \times n$ strictly lower triangular 
matrix
\be\label{bW}
    \bW = \sum_{c = 1}^{n - 1} w_c^\half E_{c + 1, c}.
\ee
By acting on it with $\bsX_H$ we get
\be\label{bsX_H_bW}
    \bsX_H[\bW]
    = \half \sum_{c = 1}^{n - 1} w_c^{-\half} \bsX_H[w_c] E_{c + 1, c}.
\ee
Note that each entry of the above matrix can be made fairly 
explicit by exploiting the equations highlighted in \eqref{bsX_H_w_1}, 
\eqref{bsX_H_w_a} and \eqref{bsX_H_w_n-1}.

So far we have encountered only with quadratic matrices of size $n$. 
However, in our theory the $(2 n) \times (2 n)$ matrices partitioned 
into four $n \times n$ blocks play a more prominent role. Accordingly, by
a $(2 n) \times (2 n)$ block matrix we shall always mean a 
$(2 n) \times (2 n)$ matrix partitioned into four blocks of equal size.
Thus, it proves handy to introduce the shorthand notation
\be\label{N}
    N = 2 n.
\ee
What is more important, based upon the $n \times n$ matrices $\bD$~\eqref{bD} 
and $\bW$~\eqref{bW}, we now construct the $N \times N$ block matrices
\be\label{D&W}
    D = \BLOCKMAT{\bD}{\bszero_n}{\bszero_n}{R \bD^{-1} R}
    \midand
    W = \BLOCKMAT{\bW}{\bszero_n}{w_n^\half E_{1, n}}{-R \bW^t R},
\ee
that allows us to introduce one of the most important objects in our study 
by the invertible lower bidiagonal matrix
\be\label{L}
    L = D (\bsone_N + W).
\ee
As will transpire, this matrix do play a distinguished role in the 
construction of a Lax pair for the deformed relativistic Toda model 
\eqref{H}. To this end, we shall need two more matrix valued functions.

First, we introduce $n$ auxiliary smooth functions by setting
\begin{align}
    & \phi_1 = \quarterbeta \BR{x_2 w_1 - x_1 - x_1^{-1}}, 
    \label{phi_1} \\
    & \phi_n = \quarterbeta 
                \BR{x_{n - 1}^{-1} w_{n - 1} 
                    - x_n 
                    - x_n^{-1} (1 - w_n) \Delta
                    + \kappa w_{n - 1} w_n},
    \label{phi_n}
\end{align}
whereas for $2 \leq a \leq n - 1$ we define
\be\label{phi_a}
    \phi_a = \quarterbeta 
                \BR{x_{a - 1}^{-1} w_{a - 1} 
                    + x_{a + 1} w_a - x_a - x_a^{-1}
                    + \kappa w_{n - 1} w_n \delta_{a, n - 1}}.
\ee 
Now, utilizing these functions, let us consider the $N \times N$ block matrix
\be\label{A}
    A = \BLOCKMAT{A^{++}}{A^{+-}}{A^{-+}}{R \hat{A}^{--} R},
\ee
where the diagonal blocks are defined by the $n \times n$ tridiagonal matrices
\begin{align}
    A^{++} =    
    & \sum_{c = 1}^{n - 1} \phi_c E_{c, c}
        + \Big( 
            \phi_n + \halfbeta \kappa w_n + \quarterbeta \kappa^2 w_n \Gamma 
        \Big) E_{n, n} \nonumber \\
    & + \halfbeta 
        \sum_{c = 1}^{n - 1}
            w_c^\half \big( x_c^{-1} E_{c, c + 1} - x_c E_{c + 1, c} \big),
    \label{A++} \\
    \hat{A}^{--} = 
    & \sum_{c = 1}^{n - 1} \phi_c E_{c, c}
        + \Big( 
            \phi_n - \halfbeta \kappa w_n -\quarterbeta \kappa^2 w_n \Gamma 
        \Big) E_{n, n} \nonumber \\ 
    & + \halfbeta 
        \sum_{c = 1}^{n - 2} 
            w_c^\half
            \big( x_{c + 1}^{-1} E_{c, c + 1} - x_{c + 1} E_{c + 1, c} \big) 
        \nonumber \\
    & + \halfbeta w_{n - 1}^\half 
            \big( x_n^{-1} \Delta E_{n - 1, n} - x_n E_{n, n - 1} \big),
    \label{hatA--}
\end{align}
whilst the off-diagonal blocks are given by the $n \times n$ spare matrices
\begin{align}
    A^{+-} =    &   \halfbeta w_n^\half 
                        \BR{ x_n^{-1} \Delta + \kappa (1 + w_{n - 1}) } 
                        E_{n, 1} 
                    + \halfbeta \kappa w_{n - 1}^\half w_n^\half E_{n, 2},
    \label{A+-} \\
    A^{-+} =    &   - \halfbeta w_n^\half 
                        \BR{ x_n + \kappa (1 + w_n) } E_{1, n}
                    + \halfbeta \kappa w_{n - 1}^\half w_n^\half E_{2, n}. 
    \label{A-+}
\end{align}

Second, in order to save space, we introduce the shorthand notations
\begin{align}
    & \psi_1 = - \quarterbeta \BR{x_1 + x_1^{-1} (1 - w_1)},
    \label{psi_1} \\
    & \psi_n = - \quarterbeta \BR{x_n (1 - w_{n - 1}) 
                                    + x_n^{-1} (1 - w_n) \Delta 
                                    - 2 \kappa w_{n - 1} w_n},
    \label{psi_n}
\end{align}
whereas for $2 \leq a \leq n - 1$ we define    
\be\label{psi_a}
    \psi_a = -\quarterbeta \BR{x_a (1 - w_{a - 1}) + x_a^{-1} (1 - w_a)}.
\ee
Making use of the above smooth functions, let us also consider the 
$N \times N$ block matrix
\be\label{B}
    B = \BLOCKMAT{B^{++}}{B^{+-}}{- R B^{+-} R}{R B^{++} R},
\ee
where the pertinent blocks are built upon the $n \times n$ matrices
\begin{align}
    B^{++} =    &   \sum_{c = 1}^n \psi_c E_{c, c}
                    + \halfbeta 
                        \sum_{c = 1}^{n - 2}
                            w_c^\half
                                \big(
                                    x_c^{-1} E_{c, c + 1}
                                    -x_{c + 1} E_{c + 1, c}
                                \big) 
                \nonumber \\
                &   + \halfbeta w_{n - 1}^\half
                        \big(
                            x_{n - 1}^{-1} E_{n - 1, n}
                            -(x_n + \kappa w_n) E_{n, n - 1}
                        \big),
    \label{B++} \\
    B^{+-} =    &   \halfbeta (x_n^{-1} \Delta + \kappa) w_n^\half E_{n, 1} 
                    + \halfbeta \kappa w_{n - 1}^\half w_n^\half E_{n, 2}.
    \label{B+-}
\end{align}
Notice that the diagonal block \eqref{B++} is tridiagonal. Therefore, by 
inspecting the off-diagonal blocks \eqref{A+-}, \eqref{A-+} and \eqref{B+-}, 
too, it is clear that in the most interesting cases characterized by the 
inequality $\kappa > 0$, both $A$~\eqref{A} and $B$~\eqref{B} are 
\emph{pentadiagonal} matrices.

Finally, to be able to formulate one of the main results of the paper, from
$D$~\eqref{D&W} and $A$~\eqref{A} we construct the $N \times N$ matrix
\be\label{cA}
    \cA = D A D^{-1}, 
\ee
which is also pentadiagonal for $\kappa > 0$.

\begin{THEOREM}\label{THEOREM:Lax_triad}
The triple $(L, \cA, B)$ obeys the equation
\be\label{triad}
    \bsX_H[L] = \cA L - L B.
\ee
\end{THEOREM}

\begin{proof}
The verification of \eqref{triad} turns out to be quite laborious, based upon 
the careful study of the $N \times N$ matrix 
\be\label{Xi_def}
    \Xi = D^{-1} \BR{ \bsX_H[L] - (\cA L - L B) }.
\ee
Remembering \eqref{L} and \eqref{cA}, Leibniz rule yields
\be\label{Xi}
    \Xi = D^{-1} \bsX_H[D] (\bsone_N + W) + \bsX_H[W] 
            - A (\bsone_N + W) + (\bsone_N + W) B.
\ee
Now, by inspecting the constituent objects, from \eqref{cD} and \eqref{D&W}
we find
\be\label{bsX_H_D&W}
    D^{-1} \bsX_H[D] = \BLOCKMAT{\cD}{\bszero_n}
                                    {\bszero_n}{-R \cD R}
    \midand
    \bsX_H[W] = \BLOCKMAT{\bsX_H[\bW]}{\bszero_n}
                            {\bsX_H[w_n^\half] E_{1, n}}{-R \bsX_H[\bW]^t R}.
\ee
Therefore, recalling the structure of $A$~\eqref{A} and 
$B$~\eqref{B}, for the block matrix decomposition of $\Xi$~\eqref{Xi}
we obtain
\be\label{Xi_block_mat}
    \Xi = \BLOCKMAT{\Xi^{++}}{\Xi^{+-}}{\Xi^{-+}}{R \hat{\Xi}^{--} R},
\ee
where the diagonal blocks are built up form the $n \times n$ matrices
\begin{align}
    \Xi^{++} = 
        & \cD + \cD \bW + \bsX_H[\bW] - A^{++} - A^{++} \bW 
        \nonumber \\
        & -w_n^\half A^{+-} E_{1, n} + B^{++} + \bW B^{++},
        \label{Xi++} \\
    \hat{\Xi}^{--} = 
        & -\cD + \cD \bW^t - \bsX_H[\bW]^t 
            - \hat{A}^{--} + \hat{A}^{--} \bW^t 
        \nonumber\\
        & + B^{++} - \bW^t B^{++} + w_n^\half E_{n, n} B^{+-} R,
    \label{hatXi--}
\end{align}
whereas the off-diagonal blocks are given by
\begin{align}
    \Xi^{+-} =  &   -A^{+-} + A^{+-} R \bW^t R + B^{+-} + \bW B^{+-},
    \label{Xi+-} \\
    \Xi^{-+} =  &   -w_n^\half R \cD E_{n, n} + \bsX_H[w_n^\half] E_{1, n} 
                    -A^{-+} - A^{-+} \bW
                \nonumber \\ 
                &   - w_n^\half R \hat{A}^{--} E_{n, n} 
                    - R B^{+-} R + R \bW^t B^{+-} R
                    + w_n^\half E_{1, n} B^{++}.
    \label{Xi-+}
\end{align}
Now our plan is obvious: we are to demonstrate that each block in 
\eqref{Xi_block_mat} vanishes. However, since we have explicit formulas
for the entries of $\cD$~\eqref{cD}, $\bW$~\eqref{bW} and 
$\bsX_H[\bW]$~\eqref{bsX_H_bW}, as well as for the entries of $A$~\eqref{A}
and $B$~\eqref{B}, our task boils down to elementary algebraic manipulations.

Starting with the study of the block $\Xi^{++}$~\eqref{Xi++}, first note that
\be\label{Xi++_term6}
    A^{+-} E_{1, n} = A^{+-}_{n, 1} E_{n, n}.
\ee
Recall also that $\cD$~\eqref{cD} is diagonal, $\bW$~\eqref{bW} has 
non-trivial entries only in its first subdiagonal, whereas matrices 
$A^{++}$~\eqref{A++} and $B^{++}$~\eqref{B++} are tridiagonal. Therefore, 
invoking the decomposition \eqref{X_decomp}, we can write
\be\label{Xi++_decomp}
    \Xi^{++} = \Xi^{++}_{-2} + \Xi^{++}_{-1} + \Xi^{++}_0 + \Xi^{++}_1,
\ee
where
\be\label{Xi++_part_-2&1}
    \Xi^{++}_{-2} = -A^{++}_{-1} \bW + \bW B^{++}_{-1}
    \midand
    \Xi^{++}_1 = - A^{++}_1 + B^{++}_1,
\ee
meanwhile the remaining two terms are given by the more complicated
expressions 
\begin{align}
    & \Xi^{++}_{-1} = \cD \bW + \bsX_H[\bW] - A^{++}_{-1} - A^{++}_0 \bW 
                        + B^{++}_{-1} + \bW B^{++}_0,
    \label{Xi++_part_-1} \\
    & \Xi^{++}_0 = \cD - A^{++}_0 - A^{++}_1 \bW 
                    - w_n^\half A^{+-}_{n, 1} E_{n, n} 
                    + B^{++}_0 + \bW B^{++}_1.
    \label{Xi++_part_0} 
\end{align}
By inspecting the simpler looking formulas displayed in 
\eqref{Xi++_part_-2&1}, from \eqref{A++}, \eqref{B++} and \eqref{bW} we find
\be
    A^{++}_{-1} \bW = \bW B^{++}_{-1} =
    -\halfbeta \sum_{c = 1}^{n - 2} 
                    x_{c + 1} w_{c + 1}^\half w_c^\half E_{c + 2, c}.
\ee
It is also clear that $A^{++}_1 = B^{++}_1$, so we deduce immediately that
\be
    \Xi^{++}_{-2} = \Xi^{++}_1 = \bszero_n.
\ee
Turning to the much more involved formula given in \eqref{Xi++_part_-1}, 
notice that for all $a \in \bN_{n - 2}$ we can write
\be
    w_a^{-\half} \Xi^{++}_{a + 1, a} =
    \cD_{a + 1, a+ 1} + w_a^{-\half} \bsX_H[w_a^\half] 
        + \halfbeta (x_a - x_{a + 1}) - \phi_{a + 1} + \psi_a,
\ee
whereas
\be
    w_{n - 1}^{-\half} \Xi^{++}_{n, n - 1} = 
        \cD_{n, n} + w_{n - 1}^{-\half} \bsX_H[w_{n - 1}^\half]
        + \halfbeta (x_{n - 1} - x_n - \kappa w_n) - A^{++}_{n, n} 
        + \psi_{n - 1}.
\ee
Thus, by combining \eqref{bsX_H_w_1}, \eqref{bsX_H_w_a} and 
\eqref{bsX_H_w_n-1} with the explicit formulas given in the first half of 
this subsection, one may easily recognize that 
\be
    \Xi^{++}_{-1} = \bszero_n.
\ee 
Finally, by inspecting the diagonal matrix $\Xi^{++}_0$~\eqref{Xi++_part_0}, 
we see that
\begin{align}
    & \Xi^{++}_{1, 1} = 
        \cD_{1, 1} - \phi_1 - A^{++}_{1, 2} w_1^\half + \psi_1 = 0,
    \label{Xi++_1_1} \\
    & \Xi^{++}_{n, n} = 
        \cD_{n, n} - A^{++}_{n, n} - w_n^\half A^{+-}_{n, 1}
        + \psi_n + w_{n - 1}^\half B^{++}_{n - 1, n} = 0, 
\end{align}
whereas for $2 \leq a \leq n - 1$ we also have
\be
    \Xi^{++}_{a, a} = 
        \cD_{a, a} - \phi_a - A^{++}_{a, a + 1} w_a^\half 
        + \psi_a + w_{a - 1}^\half B^{++}_{a - 1, a} = 0.
\ee
As a consequence, we can write $\Xi^{++}_0 = \bszero_n$.  So, in the graded
decomposition \eqref{Xi++_decomp} each homogeneous term vanishes, whence 
$\Xi^{++} = \bszero_n$.

Having grasped the idea of the proof, for the remaining three blocks we
shall highlight only the most essential steps. Turning to \eqref{hatXi--},
for the last term on the right hand side we can write
\be
    E_{n, n} B^{+-} R = B^{+-}_{n, 1} E_{n, n} + B^{+-}_{n, 2} E_{n, n - 1}.
\ee
Consequently, the graded decomposition \eqref{X_decomp} of the block 
$\hat{\Xi}^{--}$ takes the form
\be\label{hatXi--_decomp}
    \hat{\Xi}^{--} = 
        \hat{\Xi}^{--}_{-1} + \hat{\Xi}^{--}_0 
        + \hat{\Xi}^{--}_1 + \hat{\Xi}^{--}_2,
\ee
where
\be\label{hatXi--_part_-1&2} 
    \hat{\Xi}^{--}_{-1} = 
        -\hat{A}^{--}_{-1} + B^{++}_{-1} 
        + w_n^\half B^{+-}_{n, 2} E_{n, n - 1}
    \midand
    \hat{\Xi}^{--}_2 = \hat{A}^{--}_1 \bW^t - \bW^t B^{++}_1,
\ee
meanwhile
\begin{align}
    & \hat{\Xi}^{--}_0 =
        -\cD - \hat{A}^{--}_0 + \hat{A}^{--}_{-1} \bW^t + B^{++}_0
        - \bW^t B^{++}_{-1} + w_n^\half B^{+-}_{n, 1} E_{n, n}, 
    \label{hatXi--_part_0} \\
    & \hat{\Xi}^{--}_1 = 
        \cD \bW^t - \bsX_H[\bW]^t - \hat{A}^{--}_1 
        + \hat{A}^{--}_0 \bW^t + B^{++}_1 - \bW^t B^{++}_0.
    \label{hatXi--_part_1} 
\end{align}
Along the same lines as in the analysis of $\Xi^{++}$, one can easily
verify that actually each term in the decomposition \eqref{hatXi--_decomp} 
vanishes, and so $\hat{\Xi}^{--} = \bszero_n$ follows at once.

As concerns the off-diagonal block $\Xi^{+-}$~\eqref{Xi+-}, from 
\eqref{E_commut_rel}, \eqref{R&E} and \eqref{bW} we see
\be\label{Xi+-_term2&4}
    A^{+-} R \bW^t R = w_{n - 1}^\half A^{+-}_{n, 2} E_{n, 1} 
    \midand 
    \bW B^{+-} = \bszero_n.
\ee
Plugging this observation into \eqref{Xi+-}, from the matrix entries 
appearing in \eqref{A+-} and \eqref{B+-} we find immediately that
\be\label{Xi+-_OK}
    \Xi^{+-} = ( B^{+-}_{n, 1} - A^{+-}_{n, 1} 
                    + w_{n - 1}^\half A^{+-}_{n, 2} ) E_{n, 1} 
                + ( B^{+-}_{n, 2} -A^{+-}_{n, 2} ) E_{n, 2}
        = \bszero_n. 
\ee

Finally, by performing the indicated matrix multiplications in \eqref{Xi-+},
straightforward calculations lead to
\be
\begin{split}
    \Xi^{-+} =  
    & \big( w_n^\half 
            (-\cD_{n, n} - \hat{A}^{--}_{n, n} + \psi_n) + \bsX_H[w_n^\half] 
            - A^{-+}_{1, n} - B^{+-}_{n, 1} 
    \big) E_{1, n} \\
    & - \big( A^{-+}_{2, n} + w_n^\half \hat{A}^{--}_{n - 1, n}
                - B^{+-}_{n, 1} w_{n - 1}^\half \big) E_{2, n} \\ 
    & - \big( A^{-+}_{1, n} w_{n - 1}^\half + B^{+-}_{n, 2} 
                - w_n^\half B^{++}_{n, n - 1} \big) E_{1, n - 1} \\
    & - w_{n - 1}^\half (A^{-+}_{2, n} - B^{+-}_{n, 2}) E_{2, n - 1}.
\end{split}
\ee
Now, looking back to the explicit form of the matrix entries, we arrive
at $\Xi^{-+} = \bszero_n$. 

To sum up, we see that each block of $\Xi$~\eqref{Xi_block_mat} vanishes. 
Therefore, turning back to the defining equation \eqref{Xi}, the Theorem 
follows.
\end{proof}

At this point we are in a position to compare the content of 
Theorem~\ref{THEOREM:Lax_triad} with the related earlier works on 
Ruijsenaars' $n$-particle open relativistic Toda chain. In this 
respect the most important contribution is due to Suris \cite{S91}, 
who constructed certain $n \times n$ bidiagonal matrices obeying 
equations analogous to \eqref{triad}. Following Suris' terminology,
we shall call \eqref{triad} a `Lax triad'. Admittedly, the bidiagonal
matrix defined in equation (11) in \cite{S91} provided a firm
guide in constructing our lower bidiagonal matrix $L$~\eqref{L}.
However, finding the correct form of $\cA$~\eqref{cA} and $B$~\eqref{B} 
required a much greater effort on our part. Indeed, while the analogous 
matrices $A_\pm$ and $B_\pm$ given below equations (13) and (14) in 
\cite{S91} are bidiagonal, our matrices $\cA$ and $B$ are pentadiagonal 
for $\kappa > 0$. 

Due to Theorem~\ref{THEOREM:Lax_triad}, matrix $\cA$ defined in \eqref{cA} is 
a key player in our study of the deformed relativistic Toda system \eqref{H}. 
Partly for this reason it is worthwhile to make it as explicit as possible. 
To this end, by conjugating $A$~\eqref{A} with $D$~\eqref{D&W} we obtain
the block matrix form
\be\label{cA_explicit}
    \cA = \BLOCKMAT{\cA^{++}}{\cA^{+-}}{\cA^{-+}}{R \hat{\cA}^{--} R},
\ee
where
\begin{align}
    \cA^{++} =  &   \sum_{c = 1}^{n - 1} \phi_c E_{c, c}
                    + \Big( 
                        \phi_n + \halfbeta \kappa w_n 
                            + \quarterbeta \kappa^2 w_n \Gamma 
                        \Big) E_{n, n} 
                \nonumber \\
                &   + \halfbeta 
                        \sum_{c = 1}^{n - 2} 
                            w_c^\half 
                            \big( 
                                x_c^{-\half} x_{c + 1}^{-\half} E_{c, c + 1}
                                - x_c^\half x_{c + 1}^\half E_{c + 1, c} 
                            \big)
                \nonumber \\
                &   + \halfbeta w_{n - 1}^\half 
                        \big( 
                            x_{n - 1}^{-\half} x_n^{-\half} 
                                \Delta^\quarter E_{n - 1, n}
                            - x_{n - 1}^\half x_n^\half 
                                \Delta^{-\quarter} E_{n, n - 1} 
                        \big),
    \label{cA++} \\
    \hat{\cA}^{--} =    &   \sum_{c = 1}^{n - 1} \phi_c E_{c, c}
                            + \Big( 
                                \phi_n - \halfbeta \kappa w_n 
                                - \quarterbeta \kappa^2 w_n \Gamma 
                                \Big) E_{n, n} 
                        \nonumber \\
                        &   + \halfbeta 
                                \sum_{c = 1}^{n - 2} 
                                    w_c^\half 
                                    \big( 
                                        x_c^{-\half} x_{c + 1}^{-\half} 
                                            E_{c, c + 1}
                                        - x_c^\half x_{c + 1}^\half 
                                            E_{c + 1, c} 
                                    \big)
                        \nonumber \\
                        &   + \halfbeta w_{n - 1}^\half 
                                \big( 
                                    x_{n - 1}^{-\half} x_n^{-\half} 
                                        \Delta^\frac{3}{4} E_{n - 1, n}
                                    - x_{n - 1}^\half x_n^\half 
                                        \Delta^\quarter E_{n, n - 1} 
                                \big),
    \label{hatcA--}
\end{align}
whereas the off-diagonal blocks are given by
\begin{align}
    \cA^{+-} = &    \halfbeta w_n^\half \Delta^\half 
                        \BR{ 1 + \kappa x_n (1 + w_{n - 1}) \Delta^{-1} } 
                            E_{n, 1}
                    + \halfbeta \kappa x_{n - 1}^\half x_n^\half 
                        w_{n - 1}^\half w_n^\half \Delta^{-\quarter} E_{n, 2},
    \label{cA+-} \\
    \cA^{-+} =&     - \halfbeta w_n^\half \Delta^\half 
                        \BR{ 1 + \kappa x_n^{-1} (1 + w_n) } E_{1, n}
                    + \halfbeta \kappa x_{n - 1}^{-\half} x_n^{-\half} 
                        w_{n - 1}^\half w_n^\half \Delta^\quarter E_{2, n}. 
    \label{cA-+}
\end{align}

\subsection{Lax equation}\label{SUBSECTION:Lax_equation}
Heading toward the Lax representation of the dynamics governed by 
$H$~\eqref{H}, we still need some preparations. Built upon the $n \times n$ 
reversal matrix $R$~\eqref{R}, we introduce the $N \times N$ matrix
\be\label{Omega}
    \Omega = \BLOCKMAT{\bszero_n}{R}{-R}{\bszero_n}.
\ee
It is plain that $\Omega^2 = -\bsone_N$, whence $\Omega$ is invertible and
\be\label{Omega_inv}
    \Omega^{-1} = -\Omega = \Omega^T.
\ee
Note that for any $N \times N$ block matrix
\be\label{X}
    X = \BLOCKMAT{X^{++}}{X^{+-}}{X^{-+}}{X^{--}},
\ee
partitioned into four $n \times n$ blocks, we have
\be\label{Omega_conj}
    \Omega X \Omega^{-1} = \BLOCKMAT{R X^{--} R}{-R X^{-+} R}
                                    {-R X^{+-} R}{R X^{++} R}. 
\ee
As concerns the relationship between the $\bZ$-grading \eqref{decomp} and
the inner automorphism provided by the $\Omega$-conjugation, from the
above relation and \eqref{R&grading} we infer
\be\label{Omega&grading}
    \Omega (\bR^{N \times N})_k \Omega^{-1} = (\bR^{N \times N})_{-k}
    \qquad
    (k \in \bZ). 
\ee
Consequently, the $\Omega$-conjugate of any upper triangular $N \times N$ 
matrix is lower triangular, and vice versa.

To continue, consider the block matrix
\be\label{g}
    g = \BLOCKMAT{g^{++}}{\kappa w_n^\half E_{n, 1}}
                    {\kappa w_n^\half E_{1, n}}{g^{--}}
\ee
with the $n \times n$ invertible diagonal blocks
\be\label{g++&g--}
    g^{++} = \sum_{c = 1}^{n - 1} E_{c, c} + \Delta^\half E_{n, n}
    \midand
    g^{--} = R g^{++} R = \Delta^\half E_{1, 1} + \sum_{c = 2}^n E_{c, c}.
\ee
Remembering $\Delta$~\eqref{Delta}, one can check that $g$ is 
an invertible matrix with inverse
\be\label{g_inv}
    g^{-1} = \BLOCKMAT{g^{++}}{-\kappa w_n^\half E_{n, 1}}
                        {-\kappa w_n^\half E_{1, n}}{g^{--}}.
\ee
Note that both $g$ and $g^{-1}$ are tridiagonal for $\kappa > 0$.

By inspecting the action of the Hamiltonian vector field 
$\bsX_H$~\eqref{bsX_H} on the matrix valued smooth function $g$, elementary 
matrix multiplication yields
\be
\label{bsX_H_g}
\begin{split}
    \bsX_H[g] g^{-1} = 
    & \big( \Delta^\half \bsX_H[\Delta^\half] 
            - \kappa^2 w_n^\half \bsX_H[w_n^\half] \big) 
        \BLOCKMAT{E_{n, n}}{\bszero_n}{\bszero_n}{E_{1, 1}} \\
    & + \kappa \big( \Delta^\half \bsX_H[w_n^\half] 
                    - w_n^\half \bsX_H[\Delta^\half] \big)
        \BLOCKMAT{\bszero_n}{E_{n, 1}}{E_{1, n}}{\bszero_n}.
\end{split}
\ee
Recalling $\Delta$~\eqref{Delta}, we see that the diagonal
entries are actually trivial, since
\be
    \Delta^\half \bsX_H[\Delta^\half] 
        - \kappa^2 w_n^\half \bsX_H[w_n^\half]
    = \half \bsX_H[\Delta] - \half \kappa^2 \bsX_H[w_n]
    = 0.
\ee
On account of \eqref{bsX_H_w_n_OK} and \eqref{bsX_H_Delta} it is also 
immediate that
\be
    \Delta^\half \bsX_H[w_n^\half] - w_n^\half \bsX_H[\Delta^\half] 
    = \half w_n^\half \Delta^\half 
        \BR{ w_n^{-1} \bsX_H[w_n] - \Delta^{-1} \bsX_H[\Delta] }
    = \halfbeta w_n^\half \Delta^\half \Upsilon,
\ee
whence we end up with the concise expression
\be\label{bsX_H_g_OK}
    \bsX_H[g] g^{-1} = \halfbeta \kappa w_n^\half \Delta^\half \Upsilon 
                        \BLOCKMAT{\bszero_n}{E_{n, 1}}{E_{1, n}}{\bszero_n}.   
\ee
Having equipped with matrices $\Omega$~\eqref{Omega} and $g$~\eqref{g}, we
are now in a position to formulate an important technical observation.

\begin{PROPOSITION}\label{PROPOSITION:cA&B&g&Omega}
The matrix valued functions $\cA$~\eqref{cA} and $B$~\eqref{B} enjoy the
following algebraic relationships:
\be\label{cA&B&g&Omega}
    \Omega \cA \Omega^{-1} = g \cA g^{-1} + \bsX_H[g] g^{-1}
    \midand
    \Omega B \Omega^{-1} = B.
\ee
\end{PROPOSITION}

\begin{proof}
Looking back to \eqref{B} and \eqref{Omega_conj}, it is trivial that the
$\Omega$-conjugate of matrix $B$ coincides with $B$. However, verifying
the indicated relationship for $\cA$ is a bit more subtle. Nevertheless, 
remembering \eqref{R&E}, \eqref{cA_explicit}, \eqref{g} and \eqref{bsX_H_g_OK}, 
we see immediately that the $N \times N$ matrix
\be\label{Lambda}
    \Lambda = \Omega \cA \Omega^{-1} - g \cA g^{-1} - \bsX_H[g] g^{-1}
\ee
has the block matrix structure
\be\label{Lambda_block_mat}
    \Lambda = 
    \BLOCKMAT{\Lambda^{++}}{\Lambda^{+-}}
                {\Lambda^{-+}}{R \hat{\Lambda}^{--} R},
\ee
where
\begin{align}
    \Lambda^{++} = 
        & \hat{\cA}^{--} - g^{++} \cA^{++} g^{++} 
            + \kappa w_n^\half g^{++} \cA^{+-} E_{1, n} 
        \nonumber \\
        & - \kappa w_n^\half E_{n, 1} \cA^{-+} g^{++} 
            + \kappa^2 w_n E_{n, n} \hat{\cA}^{--} E_{n, n},
    \label{Lambda++} \\
    \hat{\Lambda}^{--} = 
        & \cA^{++} - g^{++} \hat{\cA}^{--} g^{++} 
            + \kappa w_n^\half g^{++} R \cA^{-+} E_{n, n} 
        \nonumber \\
        & - \kappa w_n^\half E_{n, n} \cA^{+-} R g^{++} 
            + \kappa^2 w_n E_{n, n} \cA^{++} E_{n, n},
    \label{hatLambda--}
\end{align}
whereas the off-diagonal blocks are given by
\begin{align}
    \Lambda^{+-} =  &   -R \cA^{-+} R - g^{++} \cA^{+-} g^{--} 
                        + \kappa w_n^\half g^{++} \cA^{++} E_{n, 1} 
                    \nonumber \\
                    &   - \kappa w_n^\half E_{n, n} \hat{\cA}^{--} R g^{--} 
                        + \kappa^2 w_n E_{n, 1} \cA^{-+} E_{n, 1}
                        - \halfbeta \kappa w_n^\half \Delta^\half \Upsilon 
                            E_{n, 1},
    \label{Lambda+-} \\
    \Lambda^{-+} =  &   -R \cA^{+-} R - g^{--} \cA^{-+} g^{++} 
                        + \kappa w_n^\half g^{--} R \hat{\cA}^{--} E_{n, n} 
                    \nonumber \\
                    &   - \kappa w_n^\half E_{1, n} \cA^{++} g^{++} 
                        + \kappa^2 w_n E_{1, n} \cA^{+-} E_{1, n}
                        - \halfbeta \kappa w_n^\half \Delta^\half \Upsilon 
                            E_{1, n}.
    \label{Lambda-+}
\end{align}
Keeping in mind the matrix entries of $\cA$~\eqref{cA} and $g$~\eqref{g}, 
we proceed with examining the above four $n \times n$ blocks of $\Lambda$, 
one at a time. Besides the commutation relations \eqref{E_commut_rel}, 
during the calculations we shall use frequently the relationships 
\eqref{R&E} and \eqref{RER}, too.

Starting with $\Lambda^{++}$~\eqref{Lambda++}, from \eqref{cA++} and 
\eqref{g++&g--} it is plain that
\be\label{Lambda++_2nd_term}
\begin{split}
    g^{++} \cA^{++} g^{++} =    &   \sum_{c = 1}^{n - 1} \phi_c E_{c, c} 
                                    + \cA^{++}_{n, n} \Delta E_{n, n}
                                    + \sum_{c = 1}^{n - 2} 
                                        \BR{\cA^{++}_{c, c + 1} 
                                                E_{c, c + 1} 
                                            + \cA^{++}_{c + 1, c} 
                                                E_{c + 1, c}} \\
                                &   + \cA^{++}_{n - 1, n} \Delta^\half 
                                        E_{n - 1, n}
                                    + \cA^{++}_{n, n - 1} \Delta^\half 
                                        E_{n, n - 1}.
\end{split}
\ee
Continuing with the last three terms appearing on the right hand side of 
\eqref{Lambda++}, from \eqref{cA+-}, \eqref{cA-+} and \eqref{hatcA--} we
obtain
\be\label{Lambda++_last_three_terms}
    g^{++} \cA^{+-} E_{1, n} 
    + E_{n, 1} \cA^{-+} g^{++} 
    + E_{n, n} \hat{\cA}^{--} E_{n, n}
    = \big(
        (\cA^{+-}_{n, 1} + \cA^{-+}_{1, n}) \Delta^\half 
        + \hat{\cA}^{--}_{n, n} 
    \big) E_{n, n}.
\ee
Now, recalling \eqref{cA++} and \eqref{hatcA--}, from \eqref{Lambda++_2nd_term} 
it is plain that the off-diagonal entries in \eqref{Lambda++} cancel. The first
$n - 1$ diagonal entries also trivially cancel, and we end up with
\be\label{Lambda++_OK}
    \Lambda^{++} = \big( 
                        (\hat{\cA}^{--}_{n, n} - \cA^{++}_{n, n}) \Delta
                        + \kappa w_n^\half \Delta^\half 
                            (\cA^{+-}_{n, 1} - \cA^{-+}_{1, n}) 
                    \big) E_{n, n}
    = \bszero_n.
\ee

As concerns $\hat{\Lambda}^{--}$~\eqref{hatLambda--}, we can easily cook up 
a formula for the product $g^{++} \hat{\cA}^{--} g^{++}$ simply by replacing 
symbol $\cA^{++}$ in \eqref{Lambda++_2nd_term} at each place with 
$\hat{\cA}^{--}$. Furthermore, from \eqref{cA++}, \eqref{cA+-}, \eqref{cA-+} 
and \eqref{g++&g--} we obtain
\begin{align}
    & g^{++} R \cA^{-+} E_{n, n} = 
        \cA^{-+}_{1, n} \Delta^\half E_{n, n} + \cA^{-+}_{2, n} E_{n - 1, n}, 
    \label{hatLambda--_3rd_term} \\
    & E_{n, n} \cA^{+-} R g^{++} = 
        \cA^{+-}_{n, 1} \Delta^\half E_{n, n} + \cA^{+-}_{n, 2} E_{n, n - 1},
    \label{hatLambda--_4th_term} \\
    & E_{n, n} \cA^{++} E_{n, n} = \cA^{++}_{n, n} E_{n, n}.
    \label{hatLambda--_5th_term}
\end{align}
Now, by plugging \eqref{cA++} and the above expressions into 
\eqref{hatLambda--}, the majority of the terms trivially cancel, and we are 
left with
\be\label{hatLambda--_OK}
\begin{split}
    \hat{\Lambda}^{--} =    &   \big( 
                                    (\cA^{++}_{n, n} 
                                        - \hat{\cA}^{--}_{n, n}) \Delta
                                    + \kappa w_n^\half \Delta^\half 
                                        (\cA^{-+}_{1, n} - \cA^{+-}_{n, 1}) 
                                \big) E_{n, n} \\
                            &   + \big( \cA^{++}_{n - 1, n} 
                                        - \hat{\cA}^{--}_{n - 1, n} 
                                            \Delta^\half 
                                        + \kappa w_n^\half \cA^{-+}_{2, n} 
                                    \big) E_{n - 1, n} \\
                            &   + \big( 
                                        \cA^{++}_{n, n - 1} 
                                        - \hat{\cA}^{--}_{n, n - 1} 
                                            \Delta^\half 
                                        - \kappa w_n^\half \cA^{+-}_{n, 2} 
                                    \big) E_{n, n - 1}.
\end{split}
\ee
However, from the explicit form of the matrix entries we get
$\hat{\Lambda}^{--} = \bszero_n$.

Turning to the off-diagonal block $\Lambda^{+-}$~\eqref{Lambda+-}, notice 
the following:
\begin{align}
    & R \cA^{-+} R = 
        \cA^{-+}_{1, n} E_{n, 1} + \cA^{-+}_{2, n} E_{n - 1, 1}, 
    \label{Lambda+-_1st_term} \\
    & g^{++} \cA^{+-} g^{--} = 
        \cA^{+-}_{n, 1} \Delta E_{n, 1} 
        + \cA^{+-}_{n, 2} \Delta^\half E_{n, 2},
    \label{Lambda+-_2nd_term} \\
    & g^{++} \cA^{++} E_{n, 1} = 
        \cA^{++}_{n, n} \Delta^\half E_{n, 1}
        + \cA^{++}_{n - 1, n} E_{n - 1, n},
    \label{Lambda+-_3rd_term} \\
    & E_{n, n} \hat{\cA}^{--} R g^{--} = 
        \hat{\cA}^{--}_{n, n} \Delta^\half E_{n, 1}
        + \hat{\cA}^{--}_{n, n - 1} E_{n, 2},
    \label{Lambda+-_4th_term} \\
    & E_{n, 1} \cA^{-+} E_{n, 1} = \cA^{-+}_{1, n} E_{n, 1}.
    \label{Lambda+-_5th_term}
\end{align}
Plugging the above formula into \eqref{Lambda+-}, we obtain
\be\label{Lambda+-_OK}
\begin{split}
    \Lambda^{+-} =  &   \Big( 
                            (\kappa^2 w_n - 1) \cA^{-+}_{1, n}
                            - \cA^{+-}_{n, 1} \Delta 
                            + \kappa w_n^\half \Delta^\half 
                                \Big( 
                                    \cA^{++}_{n, n} - \hat{\cA}^{--}_{n, n}
                                    - \halfbeta \Upsilon
                                \Big) 
                        \Big) E_{n, 1} \\
                    &   + \big( 
                                \kappa w_n^\half \cA^{++}_{n - 1, n}
                                -\cA^{-+}_{2, n}  
                            \big) E_{n - 1, 1} 
                        - \big( 
                                \cA^{+-}_{n, 2} \Delta^\half 
                                + \kappa w_n^\half \hat{\cA}^{--}_{n, n - 1} 
                            \big) E_{n, 2}.
\end{split}
\ee
Remembering \eqref{Upsilon} and the matrix entries of $\cA$, we find
$\Lambda^{+-} = \bszero_n$.

Finally, for the terms appearing in the formula of 
$\Lambda^{-+}$~\eqref{Lambda-+} we can write
\begin{align}
    & R \cA^{+-} R = 
        \cA^{+-}_{n, 1} E_{1, n} + \cA^{+-}_{n, 2} E_{1, n - 1}, 
    \label{Lambda-+_1st_term} \\
    & g^{--} \cA^{-+} g^{++} = 
        \cA^{-+}_{1, n} \Delta E_{1, n} 
        + \cA^{-+}_{2, n} \Delta^\half E_{2, n},
    \label{Lambda-+_2nd_term} \\
    & g^{--} R \hat{\cA}^{--} E_{n, n} = 
        \hat{\cA}^{--}_{n, n} \Delta^\half E_{1, n}
        + \hat{\cA}^{--}_{n - 1, n} E_{2, n},
    \label{Lambda-+_3rd_term} \\
    & E_{1, n} \cA^{++} g^{++} = 
        \cA^{++}_{n, n} \Delta^\half E_{1, n}
        + \cA^{++}_{n, n - 1} E_{1, n - 1},
    \label{Lambda-+_4th_term} \\
    & E_{1, n} \cA^{+-} E_{1, n} = \cA^{+-}_{n, 1} E_{1, n}.
    \label{Lambda-+_5th_term}
\end{align}
It readily follows
\be\label{Lambda-+_OK}
\begin{split}
    \Lambda^{-+} =  &   \Big( 
                            (\kappa^2 w_n - 1) \cA^{+-}_{n, 1}
                            - \cA^{-+}_{1, n} \Delta 
                            + \kappa w_n^\half \Delta^\half 
                                \Big( 
                                    \hat{\cA}^{--}_{n, n} - \cA^{++}_{n, n} 
                                    - \halfbeta \Upsilon 
                                \Big) 
                        \Big) E_{1, n} \\
                    &   - \big( 
                                \cA^{+-}_{n, 2} 
                                + \kappa w_n^\half \cA^{++}_{n, n - 1} 
                            \big) E_{1, n - 1}  
                        + \big( 
                                \kappa w_n^\half \hat{\cA}^{--}_{n - 1, n}
                                - \cA^{-+}_{2, n} \Delta^\half 
                            \big) E_{2, n},
\end{split}
\ee
which immediately leads to $\Lambda^{-+} = \bszero_n$. 

Summarizing, we see that each block of $\Lambda$~\eqref{Lambda_block_mat} 
is zero, and so necessarily $\Lambda = \bszero_N$. Thus, remembering the 
definition \eqref{Lambda}, the proof is complete.
\end{proof}

Having completed the preparations, at this point we define a \emph{Lax matrix} 
for the deformed relativistic Toda system \eqref{H} by the matrix valued 
function
\be\label{cL}
    \cL = L \Omega L^{-1} \Omega^{-1} g.
\ee
Recall that the constituent matrix $L$~\eqref{L} lower bidiagonal. Thus, on 
account of \eqref{Omega&grading}, the product $\Omega L \Omega^{-1}$ is an 
upper bidiagonal matrix, and so its inverse is upper triangular. Consequently,
\be\label{cH}
    \cH = L \Omega L^{-1} \Omega^{-1} = L (\Omega L \Omega^{-1})^{-1}
\ee
is an upper Hessenberg matrix. Since for $\kappa > 0$ the matrix 
$g$~\eqref{g} is tridiagonal, it is plain that the proposed Lax matrix
\be\label{cL&cH}
    \cL = \cH g
\ee
has lower bandwidth $2$. So, in the most interesting cases $\cL$ is not 
Hessenberg, but its strictly lower triangular part is still quite spare.

\begin{THEOREM}\label{THEOREM:Lax_eqn}
The matrix valued functions $\cL$~\eqref{cL} and $\cA$~\eqref{cA} obey the 
Lax equation
\be\label{Lax_eqn}
    \bsX_H[\cL] = [\cA, \cL].
\ee
In other words, the pair of matrices $(\cL, \cA)$ provides a Lax pair for 
the Hamiltonian dynamics generated by $H$~\eqref{H}.
\end{THEOREM}

\begin{proof}
Remembering \eqref{cL}, Leibniz rule allows us to write
\be
    \bsX_H[\cL] = \bsX_H[L] \Omega L^{-1} \Omega^{-1} g 
                    - L \Omega L^{-1} \bsX_H[L] L^{-1} \Omega^{-1} g
                    + L \Omega L^{-1} \Omega^{-1} \bsX_H[g].
\ee
Therefore, by exploiting Theorem~\ref{THEOREM:Lax_triad}, we obtain
\be
\begin{split}
    \bsX_H[\cL] =
        & \cA \cL - L \Omega L^{-1} \Omega^{-1} (\Omega \cA \Omega^{-1}) g 
            + \cL g^{-1} \bsX_H[g] \\
        & - L (B \Omega) L^{-1} \Omega^{-1} g 
            + L (\Omega B) L^{-1} \Omega^{-1} g,
\end{split}
\ee
and so by invoking Proposition~\ref{PROPOSITION:cA&B&g&Omega}, the Theorem 
follows.
\end{proof}

\section{Solution algorithm}\label{SECTION:Solution_algo}
Constructing solution algorithms from the Lax representation of the dynamics 
is an important aspect of the theory of integrable systems (for review see 
e.g. \cite{Pe}). Since this is the usual precursor for a more geometric 
treatment in the framework of symplectic reductions, we find it highly 
motivated to present a `projection method' for the van Diejen--Toda systems 
\eqref{H}, too. We achieve this goal by adapting to our deformed systems 
the algebraic machinery available for Ruijsenaars' relativistic Toda 
chains (see Section 4 in \cite{R90}). To make this approach work, in the 
first half of this section we establish some useful algebraic properties 
of the Lax pair $(\cL, \cA)$ with members given in \eqref{cL} and 
\eqref{cA_explicit}. During the calculations we shall utilize the 
$\bZ$-gradation \eqref{X_decomp} and the notations \eqref{X_triang} and 
\eqref{X_strictly_triang}, too, but this time only for $N \times N$ matrices.

\subsection{Algebraic relations}\label{SUBSECTION:alg_rels}
Recalling \eqref{cL&cH}, we start our investigation with the upper Hessenberg 
matrix $\cH$~\eqref{cH}, which is essentially built upon the lower bidiagonal 
matrix $L$~\eqref{L}. Now, giving a glance at $D$ and $W$ (see \eqref{D&W}), 
from \eqref{Omega_conj} it is plain that
\be\label{Omega&D&W}
    \Omega D \Omega^{-1} = D^{-1}
    \midand
    \Omega W \Omega^{-1} = -W^t,
\ee
whence
\be\label{L&Omega}
    \Omega L \Omega^{-1} = D^{-1} (\bsone_N - W^t).
\ee
Consequently, the definition of $\cH$~\eqref{cH} immediately leads 
to the formula
\be\label{cH&D&W}
    \cH = D (\bsone_N + W) (\bsone_N - W^t)^{-1} D.
\ee
However, since $W$~\eqref{D&W} is manifestly \emph{nilpotent}, obeying 
$W^N = \bszero_N$, the inverse appearing in the above equation can be 
calculated as
\be\label{inverse_rel}
    (\bsone_N - W^t)^{-1} = \bsone_N + W^t + (W^t)^2 + \cdots + (W^t)^{N - 1}.
\ee
Consequently, for $\cH$ we obtain the graded decomposition
\be\label{cH_decomp}
    \cH = \cH_{-1} + \cH_0 + \cH_1 + \cdots + \cH_{N - 1},
\ee
where
\be\label{cH_-1&N-1}
    \cH_{-1} = D W D
    \midand
    \cH_{N - 1} = D (W^t)^{N - 1} D,
\ee
whereas for $0 \leq k \leq N - 2$ we have
\be\label{cH_part_k}
    \cH_k 
    = D ( \bsone_N + W W^t ) (W^t)^k D.
\ee
Of course, we still need control over the diagonal matrix
\be\label{cT}
    \cT = \bsone_N + W W^t.
\ee
Straightforward calculations show that it has the block matrix structure
\be\label{cT_block_mat}
    \cT = \BLOCKMAT{\cT^{++}}{\bszero_n}{\bszero_n}{R \hat{\cT}^{--} R},
\ee
where
\be\label{cT++&hatcT--}
    \cT^{++} = E_{1, 1} + \sum_{c = 2}^n (1 + w_{c - 1}) E_{c, c}
    \midand
    \hat{\cT}^{--} = \sum_{c = 1}^n (1 + w_c) E_{c, c}.
\ee
At this point, in principle, we could easily work out explicit formulas for 
the matrix entries of $\cH$, too. However, the good news is that later on we 
shall need control over only the homogeneous parts $\cH_{-1}$, $\cH_0$ and 
$\cH_1$.

Starting with the degree $-1$ part of $\cH$~\eqref{cH_decomp}, from 
\eqref{D&W} and \eqref{cH_-1&N-1} we see that it can be partitioned into
the block matrix form
\be\label{cH_part_-1_block_mat}
    \cH_{-1} = 
    \BLOCKMAT{\BR{\cH_{-1}}^{++}}{\bszero_n}
                {\BR{\cH_{-1}}^{-+}}{\BR{\cH_{-1}}^{--}},
\ee
where the non-trivial $n \times n$ blocks are given by
\begin{align}
    & \BR{\cH_{-1}}^{++} =
        \sum_{c = 1}^{n - 2} x_c^\half x_{c + 1}^\half w_c^\half E_{c + 1, c}
        + x_{n - 1}^\half x_n^\half w_{n - 1}^\half \Delta^{-\quarter} 
            E_{n, n - 1},
    \label{cH_part_-1_++} \\
    & \BR{\cH_{-1}}^{--} =
        - R \Big( 
                \sum_{c = 1}^{n - 2} x_c^{-\half} x_{c + 1}^{-\half} 
                    w_c^\half E_{c, c + 1}
                + x_{n - 1}^{-\half} x_n^{-\half} w_{n - 1}^\half 
                    \Delta^\quarter E_{n - 1, n}
        \Big) R,
    \label{cH_part_-1_--} \\
    & \BR{\cH_{-1}}^{-+} = w_n^\half E_{1, n}.
    \label{cH_part_-1_-+} 
\end{align}
Incidentally, from the above formulas it is apparent that $\cH$ is 
actually an \emph{unreduced} upper Hessenberg matrix, meaning that all 
its entries in the first subdiagonal are nonzero. 

As concerns the diagonal part of $\cH$, from \eqref{cH_part_k} and \eqref{cT} 
we obtain
\be
    \cH_0 = D \cT D = \cT D^2.
\ee 
Since we are dealing here with diagonal matrices, we find effortlessly that
\be\label{cH_part_0_block_mat}
    \cH_0 = \BLOCKMAT{(\cH_0)^{++}}{\bszero_n}{\bszero_n}{(\cH_0)^{--}}
\ee
with blocks
\begin{align}
    & (\cH_0)^{++} = 
        x_1 E_{1, 1} 
        + \sum_{c = 2}^{n - 1} x_c (1 + w_{c - 1}) E_{c, c} 
        + x_n (1 + w_{n - 1}) \Delta^{-\half} E_{n, n},
    \label{cH_part_0_++} \\
    & (\cH_0)^{--} =
        R \Big(
            \sum_{c = 1}^{n - 1} x_c^{-1} (1 + w_c) E_{c, c}
            + x_n^{-1} (1 + w_n) \Delta^\half E_{n, n}
        \Big) R.
    \label{cH_part_0_--}
\end{align}

Finally, turning now to the degree $1$ part of $\cH$, due to 
\eqref{cH_part_k} we can write
\be
    \cH_1 = D \cT W^t D = \cT D W^t D = \cT (D W D)^t = \cT \BR{ \cH_{-1} }^t.
\ee
However, since we already have explicit formulas for the matrices $\cT$ and 
$\cH_{-1}$, it comes without any difficulty that
\be\label{cH_part_1_block_mat}
    \cH_1 = 
    \BLOCKMAT{\BR{\cH_1}^{++}}{\BR{\cH_1}^{+-}}
                {\bszero_n}{\BR{\cH_1}^{--}},
\ee
where the non-trivial $n \times n$ blocks are furnished by the slightly 
complicated expressions
\begin{align}
    \BR{\cH_1}^{++} =   
    & x_1^\half x_2^\half w_1^\half E_{1, 2}
        + \sum_{c = 2}^{n - 2} 
            x_c^\half x_{c + 1}^\half w_c^\half (1 + w_{c - 1}) E_{c, c + 1}
    \nonumber \\
    & + x_{n - 1}^\half x_n^\half 
        w_{n - 1}^\half (1 + w_{n - 2}) \Delta^{-\quarter} E_{n - 1, n},
    \label{cH_part_1_++} \\
    \BR{\cH_1}^{+-} =   
    & w_n^\half (1 + w_{n - 1}) E_{n, 1},
    \label{cH_part_1_+-}
\end{align}
whilst the lower-right-hand corner can be recovered from
\be\label{cH_part_1_--}
\begin{split}
    R \BR{\cH_1}^{--} R =   
    & - \sum_{c = 1}^{n - 2} 
                x_c^{-\half} x_{c + 1}^{-\half} w_c^\half (1 + w_{c + 1}) 
                    E_{c + 1, c} \\
    & - x_{n - 1}^{-\half} x_n^{-\half} w_{n - 1}^\half (1 + w_n) 
        \Delta^\quarter E_{n, n - 1}.
\end{split}
\ee

Related to upper Hessenberg matrix $\cH$~\eqref{cH}, we find it convenient 
to introduce the conjugated matrix
\be\label{cHtilde}
    \tilde{\cH} = \Omega \cH \Omega^{-1}.
\ee
By exploiting \eqref{Omega_conj}, it is evident that $\tilde{\cH}$ is an 
unreduced lower Hessenberg matrix, and a moment of reflection also reveals 
that
\be\label{cH_inv}
    \tilde{\cH} = \Omega L \Omega^{-1} L^{-1} = \cH^{-1}.
\ee
Thus, keeping in mind \eqref{Omega&grading}, from \eqref{cH_decomp} it 
follows that $\tilde{\cH}$ has the graded decomposition
\be\label{tildecH_decomp}
    \tilde{\cH} = 
    \tilde{\cH}_{-(N - 1)} + \cdots + \tilde{\cH}_{-1} 
    + \tilde{\cH}_0 + \tilde{\cH}_1,
\ee
where
\be
    \tilde{\cH}_k = \Omega \cH_{-k} \Omega^{-1}
    \qquad
    (-(N - 1) \leq k \leq 1).
\ee
In particular, making use of \eqref{Omega_conj}, we can easily cook up 
explicit formulas for the matrices $\tilde{\cH}_{-1}$, $\tilde{\cH}_0$ and 
$\tilde{\cH}_1$ from the above given expressions for $\cH_1$, $\cH_0$ and 
$\cH_{-1}$, respectively.

As the last piece of preparation of this subsection, with the aid of the 
standard coordinate functions \eqref{q&theta} we introduce the matrix valued 
smooth function
\be\label{Q}
    Q = \diag(q_1, \ldots, q_n, -q_n, \ldots, -q_1)
        \in C^\infty(P, \bR^{N \times N}).
\ee

\begin{LEMMA}\label{LEMMA:cL&cA&Q&H}
The matrices $\cL$~\eqref{cL} and $\cA$~\eqref{cA_explicit} obey the algebraic
relations
\be\label{cL&cA_rels}
    \Big( \cA + \halfbeta \cL \Big)_{<0} = \bszero_N
    \midand
    \Big( \cA + \halfbeta \cL^{-1} \Big)_{>0} = \bszero_N.
\ee
Furthermore, the Hamiltonian $H$~\eqref{H} and the matrix $Q$~\eqref{Q} are
related to the diagonal parts of $\cL$ and $\cL^{-1}$ by 
\be\label{H&Q&cL}
    H = \half \tr(\cL)
    \midand
    \bsX_H[Q] = \halfbeta (\cL - \cL^{-1})_0.
\ee
\end{LEMMA}

\begin{proof}
We start by noticing that $g$~\eqref{g} has the graded decomposition
\be\label{g_decomp}
    g = g_{-1} + g_0 + g_1
\ee
with homogeneous parts
\be\label{g_parts}
    g_{-1} = 
    \kappa w_n^\half \BLOCKMAT{\bszero_n}{\bszero_n}{E_{1, n}}{\bszero_n},
    \quad
    g_0 = 
    \BLOCKMAT{g^{++}}{\bszero_n}{\bszero_n}{g^{--}},
    \quad
    g_1 = 
    \kappa w_n^\half \BLOCKMAT{\bszero_n}{E_{n, 1}}{\bszero_n}{\bszero_n}.
\ee
Therefore, remembering \eqref{cL&cH} and \eqref{cH_decomp}, for the lower
triangular part of the Lax matrix we can write
\be
    \cL_{\leq 0} = \cL_{-2} + \cL_{-1} + \cL_0,
\ee
where
\be\label{cL_parts}
    \cL_{-2} = \cH_{-1} g_{-1},
    \quad
    \cL_{-1} = \cH_{-1} g_0 + \cH_0 g_{-1},
    \quad
    \cL_0 = \cH_{-1} g_1 + \cH_0 g_0 + \cH_1 g_{-1}.
\ee
Below we shall analyze each homogeneous part separately. 

Starting with the degree $-2$ part of the Lax matrix, from 
\eqref{cH_part_-1_block_mat}, \eqref{cH_part_-1_--} and \eqref{g_parts} we 
find immediately that
\be
    \cL_{-2} = 
    \kappa w_n^\half \BLOCKMAT{\bszero_n}{\bszero_n}
                                {(\cH_{-1})^{--} E_{1, n}}{\bszero_n}
    = - \kappa x_{n - 1}^{-\half} x_n^{-\half} 
        w_{n - 1}^\half w_n^\half \Delta^\quarter 
        \BLOCKMAT{\bszero_n}{\bszero_n}{E_{2, n}}{\bszero_n}.
\ee
Simply by comparing the above matrix with $\cA$~\eqref{cA_explicit}, 
focusing in particular on the form of the off-diagonal block 
$\cA^{-+}$~\eqref{cA-+}, we arrive at the conclusion
\be\label{cL&cA_part_-2}
    \cA_{-2} + \halfbeta \cL_{-2} = \bszero_N.
\ee

Proceeding with the degree $-1$ part of $\cL$, from \eqref{cL_parts} we 
infer that
\be
    \cL_{-1} = \BLOCKMAT{(\cL_{-1})^{++}}{\bszero_n}
                        {(\cL_{-1})^{-+}}{(\cL_{-1})^{--}},
\ee
where the diagonal blocks are given by
\be
    (\cL_{-1})^{++} = (\cH_{-1})^{++} g^{++}
    \midand
    (\cL_{-1})^{--} = (\cH_{-1})^{--} g^{--},
\ee
whilst the only non-trivial off-diagonal block takes the form
\be
    (\cL_{-1})^{-+} = 
        (\cH_{-1})^{-+} g^{++} + \kappa w_n^\half (\cH_0)^{--} E_{1, n}.
\ee
By performing the matrix operations, from \eqref{g++&g--}, 
\eqref{cH_part_-1_block_mat} and \eqref{cH_part_0_block_mat} we obtain
\begin{align}
    & (\cL_{-1})^{++} =
        \sum_{c = 1}^{n - 2} 
            x_c^\half x_{c + 1}^\half w_c^\half E_{c + 1, c}
        + x_{n - 1}^\half x_n^\half w_{n - 1}^\half \Delta^{-\quarter} 
            E_{n, n - 1},
    \\
    & (\cL_{-1})^{-+} =
        w_n^\half \Delta^\half \BR{ 1 + \kappa x_n^{-1} (1 + w_n) } E_{1, n},
    \\
    & (\cL_{-1})^{--} =
        -R \Big(
            \sum_{c = 1}^{n - 2} 
                x_c^{-\half} x_{c + 1}^{-\half} w_c^\half E_{c, c + 1}
            + x_{n - 1}^{-\half} x_n^{-\half} w_{n - 1}^\half 
                \Delta^{\frac{3}{4}} E_{n - 1, n}
        \Big) R.
\end{align}
Upon comparison with the corresponding blocks of 
$\cA$~\eqref{cA_explicit}, from the explicit formulas appearing in 
\eqref{cA++}, \eqref{hatcA--} and \eqref{cA-+} it is clear that
\be\label{cL&cA_part_-1}
    \cA_{-1} + \halfbeta \cL_{-1} = \bszero_N.
\ee
Now, merging this observation with \eqref{cL&cA_part_-2}, notice that the 
first relationship displayed in \eqref{cL&cA_rels} arises at once.

Turning to the diagonal part of $\cL$, from \eqref{cL_parts} we see
that
\be\label{cL_0_block_mat}
    \cL_0 = \BLOCKMAT{(\cL_0)^{++}}{\bszero_n}{\bszero_n}{(\cL_0)^{--}},
\ee
where the diagonal blocks are given by
\begin{align}
    & (\cL_0)^{++} = (\cH_0)^{++} g^{++} 
                        + \kappa w_n^\half (\cH_1)^{+-} E_{1, n},
    \label{cL_part_0_++} \\
    & (\cL_0)^{--} = \kappa w_n^\half (\cH_{-1})^{-+} E_{n, 1} 
                        + (\cH_0)^{--} g^{--}.
    \label{cL_part_0_--}
\end{align}
Now, by exploiting the equations \eqref{g++&g--}, \eqref{cH_part_-1_block_mat},
\eqref{cH_part_0_block_mat} and \eqref{cH_part_1_block_mat}, notice that the 
above diagonal matrices can be cast into the form
\begin{align}
    & (\cL_0)^{++} = 
        x_1 E_{1, 1} 
        + \sum_{c = 2}^{n - 1} x_c (1 + w_{c - 1}) E_{c, c} 
        + (x_n + \kappa w_n) (1 + w_{n - 1}) E_{n, n},
    \label{cL_part_0_++_OK} \\
    & (\cL_0)^{--} =
        R \Big(
            \sum_{c = 1}^{n - 1} x_c^{-1} (1 + w_c) E_{c, c}
            + (x_n^{-1} (1 + w_n) \Delta + \kappa w_n) E_{n, n}
        \Big) R.
    \label{cL_part_0_--_OK}
\end{align}
Therefore, taking into account Lemma~\ref{LEMMA:H}, it is immediate that
\be
    \tr(\cL) = 
    \tr(\cL_0) = 
    \tr \BR{ (\cL_0)^{++} } + \tr \BR{ (\cL_0)^{--} } = 
    2 H,
\ee
establishing the first relationship in \eqref{H&Q&cL} between the 
Hamiltonian function $H$~\eqref{H} and the Lax matrix $\cL$~\eqref{cL}.

Having completed the study of the lower triangular part of $\cL$, it is 
clear that the upper triangular part of $\cL^{-1}$ can be analyzed by the 
same technique. Highlighting only the major steps of the calculations, 
let us observe that on account of \eqref{cL&cH} and \eqref{cH_inv} we can 
write
\be
    \cL^{-1} = g^{-1} \cH^{-1} = g^{-1} \tilde{\cH}.
\ee
Remembering \eqref{g_inv} and \eqref{g_parts} we see that the inverse of $g$ 
has the graded decomposition
\be\label{g_inv_decomp}
    g^{-1} = -g_{-1} + g_0 - g_1.
\ee
Therefore, combining this observation with \eqref{cH_decomp}, it is plain 
that 
\be
    \BR{\cL^{-1}}_{\geq 0} = (\cL^{-1})_0 + (\cL^{-1})_1 + (\cL^{-1})_2,
\ee
where the diagonal part is given by
\be\label{cL_inv_part_0}
    (\cL^{-1})_0 = 
        -g_{-1} \tilde{\cH}_1 + g_0 \tilde{\cH}_0 - g_1 \tilde{\cH}_{-1},
\ee
whereas the strictly upper triangular part is built upon the homogeneous 
parts
\be\label{cL_inv_part_1&2}
    (\cL^{-1})_1 = g_0 \tilde{\cH}_1 - g_1 \tilde{\cH}_0
    \midand
    (\cL^{-1})_2 = - g_1 \tilde{\cH}_1.
\ee
Inspecting each homogeneous part separately, notice that
\be
\begin{split}
    (\cL^{-1})_2    &   = -\kappa w_n^\half 
                            \BLOCKMAT{\bszero_n}{E_{n, 1}}
                                        {\bszero_n}{\bszero_n}
                            \BLOCKMAT{R (\cH_{-1})^{--} R}
                                        {-R (\cH_{-1})^{-+} R}
                                        {\bszero_n}
                                        {R (\cH_{-1})^{++} R} \\
                    &   = -\kappa w_n^\half 
                            \BLOCKMAT{\bszero_n}
                                        {E_{n, 1} R (\cH_{-1})^{++} R }
                                        {\bszero_n}
                                        {\bszero_n}
                        = -\kappa x_{n - 1}^\half x_n^\half w_{n - 1}^\half 
                            w_n^\half \Delta^{-\quarter} 
                            \BLOCKMAT{\bszero_n}{E_{n, 2}}
                                        {\bszero_n}{\bszero_n},
\end{split}
\ee
meanwhile from \eqref{cL_inv_part_1&2} it is also clear that for the degree 
$1$ part we can write
\be
    (\cL^{-1})_1 = \BLOCKMAT{g^{++} R (\cH_{-1})^{--} R}
                            {-g^{++} R (\cH_{-1})^{-+} R}
                            {\bszero_n}
                            {g^{--} R (\cH_{-1})^{++} R}
                    - \kappa w_n^\half 
                        \BLOCKMAT{\bszero_n}{E_{n, 1} R (\cH_0)^{++} R}
                                    {\bszero_n}{\bszero_n}.
\ee
Now, recalling the matrices given in \eqref{cH_part_-1_++}, 
\eqref{cH_part_-1_--} and \eqref{cH_part_0_++}, the second relation 
displayed in \eqref{cL&cA_rels} can be confirmed easily.

Spelling out \eqref{cL_inv_part_0}, let us also observe that the diagonal 
part of $\cL^{-1}$ takes the form
\be\label{cL_inv_part_0_block_mat}
    (\cL^{-1})_0 =
        \BLOCKMAT{((\cL^{-1})_0)^{++}}{\bszero_n}
                    {\bszero_n}{((\cL^{-1})_0)^{--}},
\ee
where the non-trivial blocks are given by 
\begin{align}
    ((\cL^{-1})_0)^{++} =   &   \sum_{c = 1}^{n - 1} 
                                    x_c^{-1} (1 + w_c) E_{c, c}
                                + \BR{x_n^{-1} (1 + w_n) \Delta 
                                        + \kappa w_n (1 + w_{n - 1})} 
                                    E_{n, n},
    \label{cL_inv_part_0_++} \\
    ((\cL^{-1})_0)^{--} =   &   R \Big( 
                                        x_1 E_{1, 1} 
                                        + \sum_{c = 2}^{n - 1} 
                                            x_c (1 + w_{c - 1}) E_{c, c} 
                                        + \BR{x_n (1 + w_{n - 1}) 
                                                + \kappa w_n} E_{n, n} 
                                    \Big) R.
    \label{cL_inv_part_0_--}
\end{align}
The point is that, by subtracting the above diagonal matrices from the
corresponding diagonal matrices given in \eqref{cL_part_0_++_OK} and 
\eqref{cL_part_0_--_OK}, Lemma~\ref{LEMMA:bsX_H_q&p} immediately leads
to the second relationship in \eqref{H&Q&cL} for $\bsX_H[Q]$.
\end{proof}

\subsection{Projection method}\label{SUBSECTION:Projection_method}
In order to make the presentation simpler, it proves convenient to introduce
the shorthand notation
\be\label{cY}
    \cY = \halfbeta ( \cL - \cL^{-1} ).
\ee
Remembering Lemma~\ref{LEMMA:cL&cA&Q&H}, it is evident that for the diagonal 
part of $\cY$ we have
\be\label{cY&Q}
    \cY_0 = \bsX_H[Q].
\ee
Moreover, due to Theorem~\ref{THEOREM:Lax_eqn} we can write
\be
    \bsX_H[\cL^{-1}] 
    = - \cL^{-1} \bsX_H[\cL] \cL^{-1} 
    = - \cL^{-1} [\cA, \cL] \cL^{-1} 
    = [\cA, \cL^{-1}],
\ee
thus the derivative of $\cY$ along the Hamiltonian vector field $\bsX_H$
takes the Lax form
\be\label{cY_Lax_eqn}
    \bsX_H[\cY] = [\cA, \cY]. 
\ee

\begin{THEOREM}\label{THEOREM:Solution_algo}
Take an arbitrary point $\zeta \in P$ and consider the maximal integral curve
\be\label{gamma_zeta_in_THM}
    \gamma_\zeta \colon \bR \rightarrow P,
    \quad
    t \mapsto \gamma_\zeta(t)
\ee
of the Hamiltonian vector field $\bsX_H$~\eqref{bsX_H} satisfying the initial
condition
\be
    \gamma_\zeta(0) = \zeta.
\ee
Then one can find smooth matrix valued functions
\be
    l_\zeta \colon \bR \rightarrow \bR^{N \times N}
    \midand
    u_\zeta \colon \bR \rightarrow \bR^{N \times N},
\ee
subject to the conditions
\be
    l_\zeta(0) = u_\zeta(0) = \bsone_N,
\ee
such that for all $t \in \bR$ the matrix $l_\zeta(t)$ is unit lower 
triangular, the matrix $u_\zeta(t)$ is unit upper triangular, and most 
importantly
\be\label{exp_flow}
    e^{t \cY(\zeta)} 
    = l_\zeta(t) e^{Q(\gamma_\zeta(t)) - Q(\zeta)} u_\zeta(t).
\ee
\end{THEOREM}

\begin{proof}
Besides the algebraic developments we have made so far, the proof of this 
theorem hinges on two elementary facts from the theory of homogeneous linear 
systems of first order differential equations. First, the domain of every 
maximally defined solution of such system coincides with the common domain 
of the (continuous) coefficients. Second, if the coefficient matrix has a 
special (diagonal or triangular) form, then we can infer information about 
the structure of the fundamental matrix solution via the Peano--Baker series, 
which is basically the path-ordered matrix exponential. For a nice account 
on these facts see e.g. \cite[Corollary 3.3]{Si} and subsection 2.1.1 in 
\cite{LR}.

Continuing with the proof proper, take an arbitrary point $\zeta \in P$ and
keep it fixed. Note that, on account of the  completeness result formulated 
in Theorem~\ref{THEOREM:completeness}, the domain of the maximally defined 
integral curve \eqref{gamma_zeta_in_THM} does coincide with the set of real 
numbers. Utilizing this trajectory, below we shall introduce three 
time-dependent matrices as follows.

First, remembering the matrices $\cL$~\eqref{cL} and $\cA$~\eqref{cA}, let 
\be\label{cD_zeta}
    \cD_\zeta \colon \bR \rightarrow \bR^{N \times N},
    \quad
    t \mapsto \cD_\zeta(t)
\ee
be the (unique) maximal integral curve of the differential equation
\be\label{cD_DE}
    \dot{\cD}_\zeta(t) = - \cD_\zeta(t) 
                            \Big( 
                                \cA(\gamma_\zeta(t)) 
                                + \halfbeta \cL(\gamma_\zeta(t))^{-1} 
                            \Big)_0
\ee
with initial condition 
\be\label{cD_zeta_IC}
    \cD_\zeta(0) = \bsone_N.
\ee 
Since the coefficients in the above linear system are defined for all reals, 
the domain of the maximally defined solution $\cD_\zeta$ is indeed
$\bR$, as we anticipated in \eqref{cD_zeta}. Furthermore, since the 
coefficient matrix is diagonal, it is clear that $\cD_\zeta(t)$ is an 
\emph{invertible diagonal} matrix for all $t \in \bR$.

Second, keeping in mind $\cD_\zeta$~\eqref{cD_zeta} and $\cY$~\eqref{cY}, 
consider the maximally defined solution
\be\label{l_zeta}
    l_\zeta \colon \bR \rightarrow \bR^{N \times N},
    \quad
    t \mapsto l_\zeta(t)
\ee
of the differential equation
\be\label{l_DE}
    \dot{l}_\zeta(t) = 
        l_\zeta(t) \cD_\zeta(t) \cY(\gamma_\zeta(t))_{< 0} \cD_\zeta(t)^{-1},
\ee
satisfying the initial condition $l_\zeta(0) = \bsone_N$. 

Third, deploying $Q$~\eqref{Q}, too, let
\be\label{u_zeta}
    u_\zeta \colon \bR \rightarrow \bR^{N \times N},
    \quad
    t \mapsto u_\zeta(t)
\ee
be the maximal solution of the differential equation
\be\label{u_DE}
    \dot{u}_\zeta(t) = 
        e^{-Q(\gamma_\zeta(t)) 
        + Q(\zeta)} \cD_\zeta(t) \cY(\gamma_\zeta(t))_{> 0} 
            \cD_\zeta(t)^{-1} e^{Q(\gamma_\zeta(t)) - Q(\zeta)} u_\zeta(t),
\ee
subject to the initial condition $u_\zeta(0) = \bsone_N$.

Again, since the coefficients of the linear systems \eqref{l_DE} and 
\eqref{u_DE} are defined for all reals, so are the maximal solutions 
\eqref{l_zeta} and \eqref{u_zeta}. Notice also that the coefficient matrix 
in \eqref{l_DE} is strictly lower triangular, while the coefficient matrix 
in \eqref{u_DE} is strictly upper triangular. Though we are dealing here 
with linear systems of variable coefficients, the representation of the 
fundamental matrix by the Peano--Baker series ensures that 
$\forall t \in \bR$ the matrix $l_\zeta(t)$ is \emph{unit lower triangular}, 
whilst the matrix $u_\zeta(t)$ is \emph{unit upper triangular}. 
In particular, both $l_\zeta(t)$ and $u_\zeta(t)$ are invertible.

To proceed, for all $t \in \bR$ define the invertible lower triangular matrix
\be\label{rho}
    \rho_\zeta(t) = l_\zeta(t) \cD_\zeta(t) \in \bR^{N \times N},
\ee
and also introduce
\be\label{Phi}
    \Phi_\zeta(t) 
    = \rho_\zeta(t) \cY(\gamma_\zeta(t)) \rho_\zeta(t)^{-1}
        \in \bR^{N \times N}.
\ee
It is clear that the dependence of $\Phi_\zeta(t)$ on $t$ is smooth, and
Leibniz rule yields
\be\label{Phi_dot}
    \dot{\Phi}_\zeta(t) =
        \rho_\zeta(t) 
        \left( 
            (\cY \circ \gamma_\zeta)\spdot(t)
            + \COMM{\rho_\zeta(t)^{-1} \dot{\rho}_\zeta(t)}
                    {\cY(\gamma_\zeta(t))} 
        \right)
        \rho_\zeta(t)^{-1}.
\ee
Now, due to \eqref{cY_Lax_eqn}, for the time derivative of $\cY$ along the 
curve $\gamma_\zeta$ we can write
\be
    (\cY \circ \gamma_\zeta)\spdot(t) =
    (\bsX_H)_{\gamma_\zeta(t)}[\cY] =
    \COMM{\cA(\gamma_\zeta(t))}{\cY(\gamma_\zeta(t))},
\ee
whence \eqref{Phi_dot} entails
\be\label{Phi_dot_conj}
    \rho_\zeta(t)^{-1} \dot{\Phi}_\zeta(t) \rho_\zeta(t)
    = \COMM{\cA(\gamma_\zeta(t)) + \rho_\zeta(t)^{-1} \dot{\rho}_\zeta(t)}
            {\cY(\gamma_\zeta(t))}.
\ee
As for the derivative of \eqref{rho}, from the differential equations
\eqref{cD_DE} and \eqref{l_DE} we find
\be
\begin{split}
    \rho_\zeta(t)^{-1} \dot{\rho}_\zeta(t) 
    & = \cD_\zeta(t)^{-1} l_\zeta(t)^{-1} \dot{l}_\zeta(t) \cD_\zeta(t)
        + \cD_\zeta(t)^{-1} \dot{\cD}_\zeta(t) \\
    & = \cY(\gamma_\zeta(t))_{< 0} 
        - \Big( 
                \cA(\gamma_\zeta(t)) + \halfbeta \cL(\gamma_\zeta(t))^{-1} 
            \Big)_0.
\end{split}
\ee
Therefore, bringing into play the relations \eqref{cL&cA_rels} displayed in 
Lemma~\ref{LEMMA:cL&cA&Q&H}, we can write
\be
    \cA(\gamma_\zeta(t)) + \rho_\zeta(t)^{-1} \dot{\rho}_\zeta(t)
    = - \halfbeta \cL(\gamma_\zeta(t))^{-1}.
\ee
Plugging this formula back into \eqref{Phi_dot_conj}, from the definition
\eqref{cY} we conclude
\be
    \rho_\zeta(t)^{-1} \dot{\Phi}_\zeta(t) \rho_\zeta(t)
    = -\frac{\beta^2}{4}
        \COMM{ \cL(\gamma_\zeta(t))^{-1}}
                {\cL(\gamma_\zeta(t)) - \cL(\gamma_\zeta(t))^{-1}}
    = \bszero_N.
\ee
Now, looking back to \eqref{Phi}, it is also clear that
\be\label{Phi(0)}
    \Phi_\zeta(0) 
        = l_\zeta(0) \cD_\zeta(0) \cY(\gamma_\zeta(0)) 
            \cD_\zeta(0)^{-1} l_\zeta(0)^{-1}
        = \cY(\zeta).
\ee
The last two equations entail
\be\label{Phi_OK}
    \Phi_\zeta(t) = \cY(\zeta)
    \qquad
    (t \in \bR).
\ee
Combining this simple formula with the definition \eqref{Phi}, we end up with
\be\label{Phi_evolution}
    \cY(\gamma_\zeta(t)) = \rho_\zeta(t)^{-1} \cY(\zeta) \rho_\zeta(t)
    \qquad
    (t \in \bR).
\ee
This equation implies that the time evolution of $\cY$~\eqref{cY} 
along any trajectory is isospectral. Of course, this is what we expect from 
the Lax equation \eqref{cY_Lax_eqn}. However, the essential point here is that 
we have full control over $\rho_\zeta$~\eqref{rho} via the differential 
equations \eqref{cD_DE} and \eqref{l_DE}.

Next, for all $t \in \bR$ define the invertible matrix
\be\label{Psi}
    \Psi_\zeta(t) = l_\zeta(t) e^{Q(\gamma_\zeta(t)) - Q(\zeta)} u_\zeta(t)
        \in \bR^{N \times N}.
\ee
The dependence of $\Psi_\zeta(t)$ on $t$ is smooth, and straightforward 
calculations lead to the expression
\be
\begin{split}
    \dot{\Psi}_\zeta(t) \Psi_\zeta(t)^{-1} 
    & = \dot{l}_\zeta(t) l_\zeta(t)^{-1}
        + l_\zeta(t) \OD{\BR{Q(\gamma_\zeta(t)) - Q(\zeta)}}{t} 
            l_\zeta(t)^{-1} \\
    & \quad 
        + l_\zeta(t) e^{Q(\gamma_\zeta(t)) - Q(\zeta)} 
            \dot{u}_\zeta(t) u_\zeta(t)^{-1} 
            e^{-Q(\gamma_\zeta(t)) + Q(\zeta)} l_\zeta(t)^{-1}.
\end{split}
\ee
Now, remembering \eqref{cY&Q}, it is plain that
\be
    \OD{\BR{Q(\gamma_\zeta(t)) - Q(\zeta)}}{t} 
    = (\bsX_H)_{\gamma_\zeta(t)}[Q]
    = \cY(\gamma_\zeta(t))_0.
\ee
Thus, looking back to \eqref{l_DE}, \eqref{u_DE} and \eqref{rho}, from the
observation \eqref{Phi_evolution} we deduce
\be
\begin{split}
    \dot{\Psi}_\zeta(t) \Psi_\zeta(t)^{-1} 
    & = l_\zeta(t) 
        \big( \cD_\zeta(t) \cY(\gamma_\zeta(t))_{< 0} \cD_\zeta(t)^{-1}
                + \cY(\gamma_\zeta(t))_0 \\
    & \quad \quad \quad \quad
        + \cD_\zeta(t) \cY(\gamma_\zeta(t))_{> 0} \cD_\zeta(t)^{-1} \big)
        l_\zeta(t)^{-1} \\
    & = \rho_\zeta(t) \cY(\gamma_\zeta(t)) \rho_\zeta(t)^{-1}
        = \cY(\zeta).
\end{split}
\ee
Consequently, the function $\Psi_\zeta$~\eqref{Psi} obeys the differential 
equation
\be\label{Psi_DE}
    \dot{\Psi}_\zeta(t) = \cY(\zeta) \Psi_\zeta(t),
\ee
together with the initial condition 
\be
    \Psi_\zeta(0) 
    = l_\zeta(0) e^{Q(\gamma_\zeta(0)) - Q(\zeta)} u_\zeta(0) 
    = \bsone_N.
\ee
However, since \eqref{Psi_DE} is a homogeneous linear system with constant
coefficients, it can be solved by ordinary matrix exponentials. Indeed, 
the unique maximal solution of the above initial value problem has the form
\be\label{Psi&exp}
    \Psi_\zeta(t) 
    = e^{t \cY(\zeta)} 
    = \sum_{k = 0}^\infty \frac{t^k}{k!} \cY(\zeta)^k
    \qquad
    (t \in \bR).
\ee
Now, simply by comparing the above formula with \eqref{Psi}, the Theorem 
follows.
\end{proof}

Given an arbitrary $N \times N$ matrix 
\be\label{X_mat}
X = [X_{k, l}]_{1 \leq k, l \leq N} \in \bR^{N \times N},
\ee
for any $j \in \bN_N$ let $\pi_j(X) \in \bR$ denote its $j$th leading principal
minor; that is,
\be
    \pi_j(X) = \det([X_{k, l}]_{1 \leq k, l \leq j}).
\ee
Also, introduce the notations
\be\label{m_j}
    m_1(X) = \pi_1(X) = X_{1, 1}
    \midand
    m_j(X) = \frac{\pi_j(X)}{\pi_{j - 1}(X)}
    \qquad
    (2 \leq j \leq N).
\ee
Utilizing the above objects, we can draw important conclusions from 
the above Theorem. Indeed, the most important observation displayed
in \eqref{exp_flow} can be interpreted by saying that the matrix 
$e^{t \cY(\zeta)}$ has an $LDU$ factorization (or Gauss decomposition) 
for all $t \in \bR$. As is known from the theory of matrices (see e.g. 
\cite[Corollary 3.5.6]{HJ}), the factors appearing on the right hand 
side of \eqref{exp_flow} are unique and for the $c$th diagonal entry 
$(c \in \bN_n)$ of the diagonal factor $e^{Q(\gamma_\zeta(t)) - Q(\zeta)}$ 
we can write that
\be
    0 
    < e^{q_c(\gamma_\zeta(t)) - q_c(\zeta)} 
    = m_c \big( e^{t \cY(\zeta)} \big).
\ee
Thus, for the time evolution of the particle positions we obtain
\be
    q_c(\gamma_\zeta(t)) 
    = q_c(\zeta) + \ln\big( m_c \big( e^{t \cY(\zeta)} \big) \big),
\ee
whereas the time evolution of the rapidities can be recovered from the
relationships \eqref{bsX_H_q_1__in_q&theta}, \eqref{bsX_H_q_n__in_q&theta} 
and \eqref{bsX_H_q_a__in_q&theta}. To sum up, we see that finding the 
trajectories of the Hamiltonian dynamics generated by $H$~\eqref{H} boils 
down to the computation of the leading principal minors of the exponential 
matrix flow $t \mapsto e^{t \cY(\zeta)}$. In this sense our solution 
algorithm is of purely algebraic nature.

Note that the above discussion nicely harmonizes with Theorem 4.2 in 
\cite{R90}. Furthermore, similarly to the translation invariant case,
in our analysis we made critical use of the algebraic properties of 
the Lax pair, as formulated in Lemma~\ref{LEMMA:cL&cA&Q&H}. As a matter 
of fact, the real reason behind our choice of gauge for the Lax pair
is to maintain these algebraic relationships.

\section{Discussion}\label{SECTION:Discussion}
The main achievement of our work is the construction of a Lax representation 
for the dynamics generated by the van Diejen--Toda Hamiltonian \eqref{H}. 
By exploiting the algebraic properties of the proposed Lax pair, we could 
provide a solution algorithm, too. By the very nature of the subject, the 
presentation of the material has an inevitable algebraic flavor. However, 
even the algebraic part of the story may have further surprises in store. 
The Lax representation of the dynamics guarantees that the spectral invariants 
of Lax matrix $\cL$~\eqref{cL} are first integrals of the dynamics, but we 
are still in debt to prove that they are in involution. In the light of the 
earlier developments on Ruijsenaars' relativistic Toda chains, the most 
satisfactory step would be to provide an appropriate $r$-matrix structure 
for $\cL$. Indeed, for the translation invariant models Suris succeeded in 
casting the tensorial Poisson bracket of the Lax matrix into a Sklyanin 
bracket form. Of course, this result was instrumental in developing a 
geometric picture, too, in the framework of the Poisson--Lie groups 
(see \cite{S91}). 

Turning to questions requiring analytic considerations, it would 
be highly desirable to construct action-angle variables to the systems 
\eqref{H}. Besides Ruijsenaars' pioneering work \cite{R90}, in this respect
the paper \cite{CKA} also warrants mention, in which Suris' bidiagonal 
matrices featuring the Lax triad are directly utilized to construct 
action-angle map for the relativistic open Toda chains. Just as in 
\cite{R90}, solving this problem for our deformed systems could shed 
light on their scattering properties, too. Moreover, taking the lead of 
\cite{R90}, it would be an equally important task to uncover the dual 
systems associated with \eqref{H} in the sense of Ruijsenaars. Since the 
clarification of the above problems is unavoidable to complete the study 
of the dynamical systems \eqref{H}, we wish to come back to these issues 
in later publications.

However, the real challenge would be the construction of Lax 
matrices for the most general van Diejen--Toda systems. In this respect 
we must mention the closely related Ruijsenaars--Schnei\-der 
models \cite{RS, R88}, too. Indeed, one of the main themes of \cite{R90} 
is the transition from the hyperbolic Ruijsenaars--Schneider model to the 
relativistic Toda chain by taking the so-called `strong coupling limit'. 
In fact, the multi-parametric van Diejen--Toda chains were derived 
in a completely analogous manner from the elliptic and the hyperbolic 
van Diejen systems \cite{D94}. As concerns the construction of Lax
matrices for the van Diejen--Toda models along the same lines, for about 
two decades after their inception the main obstacle had been the very 
limited knowledge about the Lax representation of the classical van Diejen 
systems. Nevertheless, the situation has greatly improved in the last couple 
of years. In our paper \cite{P12} we provided a Lax representation for the 
rational van Diejen models with the maximal number of $3$ parameters. 
Built upon this development, we worked out a complete theory for certain 
$2$-parameter subfamily of hyperbolic van Diejen systems \cite{PG, P18}, 
too. However, in this research domain the real breakthrough is due to 
Chalykh. Indeed, among many other fascinating results, in the beautiful 
recent paper \cite{Chal} a quantum Lax matrix is constructed for the 
elliptic van Diejen system with $9$ coupling parameters.

At this point we must mention that our experience with the van Diejen 
models proved to be essential related to this paper, too. In fact, we 
arrived at our Lax matrix $\cL$~\eqref{L} by a laborious brute-force 
approach based upon a close inspection of a conjectured Lax matrix for 
certain $3$-parameter subfamily of hyperbolic van Diejen systems (see 
equation (6.5) in \cite{PG}). Fortunately, by borrowing ideas mainly 
from \cite{R90, S90}, our effort resulted in the Lax pair 
$(\cL, \cA)$ having nice properties: $\cL$ has lower bandwidth $2$, 
whereas $\cA$ is pentadiagonal for $\kappa > 0$. Therefore, it is an 
easy guess that such structured matrices will play role in the Lax 
representation of the most general van Diejen--Toda systems, too. 
However, remembering vividly our struggle with the $1$-parameter subfamily 
of deformed relativistic Toda systems \eqref{H}, we do not advocate the 
idea of a trial and error approach to construct Lax matrices with more 
coupling constants. Rather, a systematic approach is required, with more 
insight. As mentioned above, in this respect Chalykh's paper \cite{Chal} 
may come to our salvation. Indeed, by taking appropriate `strong coupling 
limits' of his Lax matrices, we expect that Lax representation will emerge 
for the most general van Diejen--Toda systems as well.

\medskip
\noindent
\textbf{Acknowledgments.}
Our work was supported by the J\'anos Bolyai Research Scholarship of the 
Hungarian Academy of Sciences, and by the \'UNKP-18-4 and the \'UNKP-19-4 
New National Excellence Programs of the Ministry of Human Capacities, Hungary. 
The support of the Ministry of Human Capacities, Hungary, by grant 
TUDFO/47138-1/2019-ITM is also greatly acknowledged; we wish to thank 
L.~Moln\'ar for the membership in his research group working on the project 
``Structures of matrices and operators and their applications''.



\begin{thebibliography}{99}

    \bibitem[BR88]{BR88}
        M.~Bruschi, O.~Ragnisco,
        Recursion operator and B\"{a}cklund transformations for the 
        Ruijsenaars--Toda lattice,
        \emph{Phys. Lett. A} \textbf{129} (1988) 21-25.

    \bibitem[BR89]{BR89}
        M.~Bruschi, O.~Ragnisco,
        Lax representation and complete integrability for the periodic 
        relativistic Toda lattice, 
        \emph{Phys. Lett. A} \textbf{134} (1989) 365-370.
        
    \bibitem[Chal]{Chal}
        O.~Chalykh,
        Quantum Lax Pairs via Dunkl and Cherednik Operators,
        \emph{Commun. Math. Phys.} \textbf{369} (2019) 261-316.
        
    \bibitem[Cher]{C}
        I.~Cherednik,
        Whittaker limits of difference spherical functions, 
        \emph{Int. Math. Res. Not.} \textbf{2009} (2009) 3793-3842.

    \bibitem[CKA]{CKA}
        J.~Coussement, A.B.J.~Kuijlaars, W.~Van~Assche,
        Direct and inverse spectral transform for the relativistic Toda 
        lattice and the connection with Laurent orthogonal polynomials,
        \emph{Inverse Problems} \textbf{18} (2002) 923-942.

    \bibitem[D94]{D94}
        J.F.~van Diejen,
        Deformations of Calogero--Moser systems and finite Toda chains,
        \emph{Theor. Math. Phys.} \textbf{99} (1994) 549-554. 

    \bibitem[D95]{D95}
        J.F.~van Diejen,
        Difference Calogero--Moser systems and finite Toda chains,
        \emph{J. Math. Phys.} \textbf{36} (1995) 1299-1323.

    \bibitem[DE]{DE}
        J.F.~van Diejen, E.~Emsiz,
        Integrable Boundary Interactions for Ruijsenaars' 
        Difference Toda Chain,
        \emph{Commun. Math. Phys.} \textbf{337} (2015) 171-189.

    \bibitem[E]{E}
        P.~Etingof, 
        Whittaker functions on quantum groups and $q$-deformed Toda operators. 
        In: 
        A.~Astashkevich, S.~Tabachnikov (eds.), 
        Differential Topology, Infinite-Dimensional Lie Algebras, 
        and Applications, 
        \emph{Amer. Math. Soc. Transl. Ser.} 2, vol. 194., 
        pp. 9-25, Amer. Math. Soc., Providence, RI, 1999.

    \bibitem[HJ]{HJ}
        R.A.~Horn, C.R.~Johnson,
        \emph{Matrix Analysis}, 2nd Ed.,
        Cambridge University Press, Cambridge, 2013.

    \bibitem[KMZ]{KMZ}
        S.~Kharchev, A.~Miranov, A.~Zhedanov, 
        Faces of relativistic Toda chain,
        \emph{Int. J. Mod. Phys. A} \textbf{12} (1997) 2675-2724.

    \bibitem[KT]{KT}
        V.B.~Kuznetsov, A.V.~Tsyganov, 
        Quantum relativistic Toda chains,
        \emph{J. Math. Sci.} \textbf{80} (1996) 1802-1810.

    \bibitem[LR]{LR}
        H.~Logemann, E.P.~Ryan,
        \emph{Ordinary Differential Equations.
        Analysis, Qualitative Theory and Control},
        Springer Undergraduate Mathematics Series,
        Springer-Verlag London, 2014.
        
    \bibitem[Pe]{Pe}
        A.M.~Perelomov,
        \emph{Integrable Systems of Classical Mechanics and Lie Algebras}, 
        Vol. 1.,
        Birkh\"auser Verlag, Boston, MA, 1990.

    \bibitem[Pu12]{P12}
        B.G.~Pusztai,
        The hyperbolic $BC_n$ Sutherland and the rational $BC_n$ 
        Ruijsenaars--Schnei\-der--van Diejen models: Lax matrices and duality,
        \emph{Nucl. Phys.} \textbf{B 856} (2012) 528-551.

    \bibitem[Pu18]{P18}
        B.G.~Pusztai, 
        Self-duality and scattering map for the hyperbolic van Diejen
        systems with two coupling parameters (with an appendix by 
        S. Ruijsenaars), 
        \emph{Commun. Math. Phys.} \textbf{359} (2018) 1-60. 
        
    \bibitem[PG]{PG}
        B.G.~Pusztai, T.F.~G\"orbe,
        Lax representation of the hyperbolic van Diejen dynamics 
        with two coupling parameters,
        \emph{Commun. Math. Phys.} \textbf{354} (2017) 829-864.

    \bibitem[R88]{R88}
        S.N.M.~Ruijsenaars, 
        Action-angle maps and scattering theory for some finite dimensional 
        integrable systems I. The pure soliton case,
        \emph{Commun. Math. Phys.} \textbf{115} (1988) 127-165.

    \bibitem[R90]{R90}
        S.N.M.~Ruijsenaars,
        Relativistic Toda Systems,
        \emph{Commun. Math. Phys.} \textbf{133} (1990) 217-247.

    \bibitem[RS]{RS}
        S.N.M.~Ruijsenaars, H.~Schneider,
        A new class of integrable models and its relation to solitons,
        \emph{Ann. Phys. (N.Y.)} \textbf{170} (1986) 370-405.

    \bibitem[Se]{S}
        A.~Sevostyanov, 
        Quantum deformation of Whittaker modules and the Toda lattice,
        \emph{Duke Math. J.} \textbf{105} (2000) 211-238.

    \bibitem[Si]{Si}
        T.C.~Sideris,
        \emph{Ordinary Differential Equations and Dynamical Systems},
        Atlantis Studies in Differential Equations, Vol. 2.,
        Atlantis Press, 2013.
        
    \bibitem[Su90]{S90}
        Yu.B.~Suris,
        Discrete time generalized Toda lattices: complete integrability and 
        relation with relativistic Toda lattices,
        \emph{Phys. Lett. A} \textbf{145} (1990) 113-119.

    \bibitem[Su91]{S91}
        Yu.B.~Suris,
        Algebraic structure of discrete-time and relativistic Toda lattices,
        \emph{Phys. Lett. A} \textbf{156} (1991) 467-474.
        
    \bibitem[Su96]{S96}
        Yu.B.~Suris,
        A discrete-time relativistic Toda lattice,
        \emph{J. Phys. A: Math. Gen.} \textbf{29} (1996) 451-465.
    
    \bibitem[Su97]{S97}
        Yu.B.~Suris,
        New integrable systems related to the relativistic Toda lattice,
        \emph{J. Phys. A: Math. Gen.} \textbf{30} (1997) 1745-1761.

    \bibitem[Su03]{S03}
        Yu.B.~Suris,
        The Problem of Integrable Discretization: Hamiltonian Approach. 
        Progress in Mathematics, vol. 219., Birk\"auser Verlag, Basel, 2003.
        
    \bibitem[Su18]{S18}
        Yu.B.~Suris, 
        Discrete time Toda systems. 
        \emph{J. Phys. A: Math. Gen.} \textbf{51} (2018) 333001. 

\end{thebibliography}
\end{document}